\documentclass[letterpaper,10pt,journal,final,onecolumn]{IEEEtran}  
 
\usepackage{amssymb,amsmath,amsthm}

\usepackage[noadjust]{cite}
\usepackage{caption}
\usepackage{subcaption}
\usepackage{float,tikz,xcolor}
\usetikzlibrary{arrows,positioning,calc}
\usetikzlibrary{arrows.meta}

\definecolor{dblue}{rgb}{0,0,.5}

\newtheorem{assumption}{Assumption}
\newtheorem{lemma}{Lemma}
\newtheorem{proposition}{Proposition}
\newtheorem{remark}{Remark}
\newtheorem{theorem}{Theorem}

\newcommand{\QED}{$\blacksquare$}
 \def\endproof{\hspace*{\fill}~\QED\par\unskip}

\newcommand{\R}{\mathbb R}
\newcommand{\C}{\mathbb C}

\newcommand{\N}{\mathbb N}

\newcommand{\Rplus}{\R_{\ge 0}}

\newcommand{\B}{\mathbb{B}}

\newcommand{\cA}{\mathcal A}

\newcommand{\cC}{\mathcal C}

\newcommand{\cK}{\mathcal K}
\newcommand{\cL}{\mathcal L}
\newcommand{\cM}{\mathcal M}

\newcommand{\cO}{\mathcal O}
\newcommand{\cP}{\mathcal P}

\newcommand{\cR}{\mathcal R}

\newcommand{\x}{\times}

\newcommand{\und}{\underline}
\newcommand{\inv}{^{-1}}

\DeclareMathOperator{\col}{col}

\DeclareMathOperator{\diag}{diag}

\newcommand{\ball}{\B}
\allowdisplaybreaks

\newcommand{\sr}{^{\star}}

\newcommand{\dw}{{n_w}}
\newcommand{\dx}{{n_x}}
\newcommand{\dy}{{n_y}}
\newcommand{\de}{{n_e}}

\newcommand{\du}{{n_u}}
\newcommand{\dya}{{n_a}}

\newcommand{\deta}{{d \de}}

\newcommand{\xb}{\mathbf{x}}
\newcommand{\pb}{\mathbf{p}}

\newcommand{\suchthat}{\,:\,}

\newcommand{\sn}{_{[0,t)}}

\title{
Output Regulation by Postprocessing Internal Models for a Class of Multivariable Nonlinear Systems
}
\author{Michelangelo~Bin~and~Lorenzo Marconi}%%

\begin{document}

  \onecolumn 
  \vspace{4em}

  \vspace{5em}
  \begin{quote}
  	\emph{This is the peer reviewed version of the following article:\emph{ M. Bin and L. Marconi,  Output regulation by postprocessing internal models for a class of multivariable nonlinear systems, Int J Robust Nonlinear Control, vol. 30, pp. 1115-1140, 2020}, which has been published in final form at 	https://doi.org/10.1002/rnc.4811. This article may be used for non-commercial purposes in accordance with Wiley Terms and Conditions for Use of Self-Archived Versions. This article may not be enhanced, enriched or otherwise transformed into a derivative work, without express permission from Wiley or by statutory rights under applicable legislation. Copyright notices must not be removed, obscured or modified. The article must be linked to Wiley’s version of record on Wiley Online Library and any embedding, framing or otherwise making available the article or pages thereof by third parties from platforms, services and websites other than Wiley Online Library must be prohibited.}
  \end{quote}

 \pagestyle{empty}

 \clearpage
 \pagestyle{plain}
 \pagenumbering{arabic} 

\maketitle

\begin{abstract}
In this paper we propose a new design paradigm, which employing a postprocessing internal model unit, to approach the problem of  output regulation for a class of multivariable minimum-phase nonlinear systems possessing a partial normal form. Contrary to previous approaches, the proposed regulator  handles control inputs of  dimension larger than the number of regulated variables, provided that a controllability assumption holds, and can employ additional   measurements that need not to vanish at the ideal error-zeroing steady state, but that can be useful for stabilization purposes  or to fulfil the minimum-phase requirement. Conditions for practical and asymptotic output regulation are given, underlying how in postprocessing schemes the design of internal models  is necessarily intertwined with that of the stabilizer.
\end{abstract}

\section{Introduction}
Output regulation is the branch of control theory studying the design of control systems making the plant follow some desired reference trajectories while rejecting at the same time unknown disturbances. The output regulation problem for linear systems was elegantly solved in the mid 70s in the seminal works of Francis, Wonham and Davison \cite{Francis1976,Davison1976} under the assumption that   all the exogenous signals, i.e. references and disturbances, are generated by a known autonomous linear process, called the \textit{exosystem}. The key result presented in those works was that a necessary \cite{Francis1976} and sufficient \cite{Davison1976} condition for a regulator to solve the linear output regulation problem \textit{robustly} (i.e. despite model uncertainties), is that the regulator must suitably embed, in the control loop, an \textit{internal model} of the exosystem.
While the concept of internal model led to a definite answer to the linear regulation problem, the situation is quite different when nonlinear systems are concerned \cite{Byrnes2003}, and nonlinear output regulation is still a quite open problem \cite{Bin2018b}. Without restrictive \textit{immersion} assumptions \cite{Huang1994,Byrnes1997,Huang2001,Huang2004,Isidori2012}, indeed, the knowledge of the exosystem alone is neither sufficient nor necessary for the solvability of the problem \cite{Byrnes2003,Bin2018b,Marconi2007}, and its role in conditioning the asymptotic behavior of the regulator mixes up with the plant's residual dynamics, thus making the celebrated robustness property of the linear regulator hard to imagine in a general nonlinear context \cite{Bin2018c}.

Under a control design perspective, nonlinear output regulation has reached a mature state, and many regulators have been proposed in the last decades  \cite{Byrnes1997,Huang2004,Marconi2007,Byrnes2004,Chen2005,IsidoriBookNew}. Nevertheless,  the existing approaches that can guarantee an asymptotically exact regulation mostly remain  limited   to minimum-phase single-input-single-output (partial) normal forms, and their immediate  square  multivariable extensions \cite{McGregor2006,Astolfi2013,Wang2016,Wang2017} (i.e. having the same number of inputs and regulation errors), where the only plant's measurements exploitable by the regulator are the \textit{regulation errors} themselves. The source of such limitation was recently sought in the common structure shared by the largest part of the nonlinear designs, which is somewhat complementary to those possessed by the original linear regulator of Davison~\cite{Davison1976}, and which induces some conceptual problems when extensions are sought. 
The linear regulator of \cite{Davison1976} is obtained by first augmenting the plant with a properly defined \textit{internal model unit},   which is  driven by the regulation errors, and then by closing the loop by means of a \textit{stabilizer}, which employs all the available measurements to stabilize the resulting closed-loop system. Since the internal model directly processes the regulation errors, this design paradigm is called \textit{postprocessing}. Most of the nonlinear approaches, instead, show a complementary design paradigm, where the internal model unit is driven by the control input produced by the stabilizer. A block-diagram representation of the two control structures is depicted in Figure \ref{fig:prepost}.
\begin{figure*}[t]
	\centering
		\tikzstyle{noshape} = [inner sep=0]
	\tikzstyle{sys} = [draw,inner sep=8]
	\tikzstyle{sum} = [circle,draw,inner sep=1]
	\tikzstyle{line} = [draw, -latex']
	\begin{tikzpicture}
	\node[sys] (plant) at (0,0) {Plant};
	\node[noshape] at (-5em,2em) {(a)};
	\node[sys] (intmod) [right=4em of plant] {Int. Model};
	\node[sys] (v) [below=2em of plant] {Stabilizer};
	\path[line]  (plant.east) -- (intmod.west) node[above,pos=.7] {$e$} ;
	\path[line] ($(plant.east)!.5!(intmod.west)$) |- ($(v.east)+(0,0.2)$) node[right,pos=.3] {$y$};
	\path[line] (v.west)--($(v.west)+(-1em,0)$) |- (plant.west) node[midway,left,pos=.25] {$u$};
	\path[line] (intmod.east)  -- ($(intmod.east)+(1em,0)$) |- ($(v.east)+(0,-0.2)$) node[right,midway,pos=.3] {$\eta$};

	\node[sys] (intmod2) at (21em,0) {Int. Model};
	\node[noshape] at (16em,2em) {(b)};
	\node[sum] (plus2) [right=2em of intmod2.east] {$\ \ \ $};
	\node[sys] (plant2) [right=2em of plus2] {Plant};
	\node[sys] (v2) [below=2em of plant2] {Stabilizer};
	\path[line]  (intmod2.east) -- (plus2.west) node [above,midway] {$\eta$};
	\path[line] (plus2.east) -- (plant2.west) node[above,midway] {$u$};
	\path[line] (plant2.east) |-   ($(plant2.east)+(2em,0)$) node [above,xshift=-10pt] {$e$} |- (v2.east) ;
	\path[line] (v2.west) node [above,xshift=-10pt] {$v$} -|  (plus2.south);
	\path[line] (v2.west) -| ($(intmod2.west)+(-1em,0)$) -- (intmod2.west);
	\end{tikzpicture} 
	\caption{(a) Postprocessing  and  (b) Preprocessing Internal Models.}
	\label{fig:prepost}
\end{figure*}
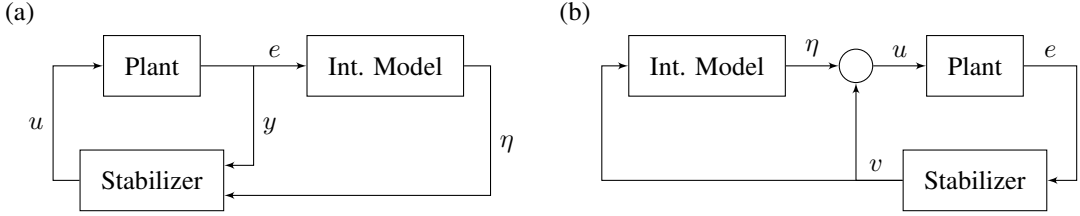
A part from their structural differences, pre and postprocessing schemes
also differ in terms of ``design philosophy'': in the postprocessing paradigm,  the plant is augmented with the internal model unit, the stabilizer is designed to stabilize the resulting cascade, so as to guarantee that the closed-loop system has a well-defined steady state (not fixed in advance) and, finally, the properties of the steady-state regulation error are inferred by the structure of the internal model. In preprocessing regulators,
instead, the internal model unit is designed in advance to be able to generate the ideal steady-state control action needed to keep the regulation errors to zero (previously computed by exploiting the known structure of the plant), and the stabilizer is then designed to ensure the asymptotic stability of such ideal steady state. Thus, while in preprocessing schemes the ideal steady states of the internal model unit and of the stabilizer are fixed a priori on the basis of the plant's data, in postprocessing
regulators  the ideal steady state for the internal model unit
and for the stabilizer cannot be fixed a priori. As a matter of fact, since the state of the internal model unit is used for stabilization purposes, and  it is thus processed by the stabilizer, its ideal steady state is strongly dependent on the choice of the stabilizer itself.

The preprocessing schemes have the interesting property that the roles of the
internal model unit and the stabilizer are neatly separated, and the ideal steady state of the closed-loop system is given by the problem data. This larger conceptual simplicity is, perhaps, the reason why most of the existing designs are of this kind. Nevertheless, preprocessing regulators have some
structural limitations that prevent their extension to larger classes of systems of those mentioned before. In particular, it is not clear,
at a conceptual level, how a preprocessing regulator could handle in a systematic way additional measured outputs that are necessary to obtain closed-loop stability (or even minimum-phase), but that need not to vanish at the ideal error-zeroing steady state, unless filtering them out at the steady state employing \textit{redundant} internal models, as recently proposed by Wang et al.~\cite{Wang2018}.  On the other hand, there is not even a clear road map on how to handle systems having more inputs than errors. If more inputs than errors are present, indeed, a preprocessing solution would lead to the employment of a number of internal models equal to the input dimension, thus yielding a redundant design, not stabilizable by error feedback (as a simple linear example would show). 
 
These conceptual problems, in principle not present in regulators of the postprocessing type, recently motivated the community to look for postprocessing alternatives to the existing regulators. In particular, the fundamental designs of Byrnes and Isidori~\cite{Byrnes2004} and Marconi et al.~\cite{Marconi2007} were ``shifted'' to an equivalent postprocessing design~\cite{Isidori2012b,Bin2017ifac}. Nevertheless, no conceptual progress has been made in terms of extensions to larger classes of systems compared to their preprocessing counterparts. A different approach to the design of postprocessing regulators was recently pursued by Astolfi  et al.~\cite{Astolfi2015b,Astolfi2017}, where the
linear regulator is attached to a class of nonlinear systems. In particular, the authors have shown that the output regulation problem can be solved \textit{robustly}~\cite{Bin2018CDC} by a postprocessing   integral action whenever the steady state is made of equilibria~\cite{Astolfi2017}, and then Astolfi et al.~\cite{Astolfi2015b} have extended the results to the case in which the steady-state signals are periodic, obtaining, however, only an approximate result stating that the Fourier coefficients
in the regulation errors corresponding to the frequencies embedded in the
internal model vanish at the steady state.

In this paper, for a class of nonlinear systems possessing a partial normal form, we investigate the existence of a postprocessing regulator  handling additional nonvanishing measurements and  input dimensions larger than the number of regulation errors. We give  conditions for asymptotic, practical, and approximate regulation results  of the same kind of those proved~\cite{Isidori2012} for the ``regression-like'' preprocessing regulator originally proposed by Byrnes and Isidori~\cite{Byrnes2004}. The proposed design shows a feature that is characteristic of postprocessing schemes (even if hidden by linearity in the linear case): the design of the stabilizer and that of the internal model unit are intertwined, and the two modules need to be co-designed (this property is known as the \textit{chicken-egg dilemma}  of output regulation~\cite{Bin2018b}). Compared to previous multivariable  approaches~\cite{Astolfi2013,Wang2016,Wang2017}, other than handling non-vanishing outputs and additional inputs, we develop the result without the quite restrictive assumption of the existence of a single-valued steady-state map defining the attractor of the zero dynamics, and we thus collocate in a broader ``nonequilibrium'' setting~\cite{Byrnes2003}.
The proposed design thus enlarges the class of multivariable systems for which the output regulation problem can be solved. Nevertheless, as  main drawbacks, the proposed approach  still limits to a design procedure strongly based on a high-gain perspective and on a minimum-phase assumption with respect to the full set of measured outputs, although the latter requirement is mitigated by the availability of additional measurements.
 Furthermore, due to the chicken-egg dilemma mentioned before,  the proposed conditions for asymptotic regulation become easily not constructive when the complexity of the problem increases. Nevertheless, the ``structure'' of the regulator remains fixed, and approximate, possibly practical, regulation is  guaranteed.\\

\textit{Notation.}  $\R$ denotes the set of real numbers, $\N$ the set of naturals and $\Rplus:=[0,\infty)$. 
With $r>0$, $\ball_r(\R^n)$ denotes the open ball of radius $r$ on $\R^n$. If $n$ is clear, we omit $\R^n$ and we write $\ball_r$. If $S\subset\R^n$ ($n\in\N$) is a set, we denote by $S^\circ$ its interior, by $\overline{S}$ is closure, and we let $|S|:=\sup_{s\in S}|s|$. With $x=(x_1,\dots,x_n)\in\R^n$ and $1\le i<j\le n$, we let $x_{[i,j]}:=(x_i,\dots,x_j)$. If $A_1,\dots,A_n$ are matrices, we let $\col(A_1,\dots,A_n)$ and $\diag(A_1,\dots,A_n)$ be their column and block-diagonal concatenations, whenever they make sense.
If $S$ is a closed set, $|x|_S:=\inf_{s\in S}|x-s|$ denotes the distance of $x\in\R^n$ to $S$. If $f$ is a function defined on $\R^n$, $f|_S$ denotes the restriction of $f$ on $S$. $\cC^1$ denotes the set of continuously differentiable functions. A continuous function $f:\Rplus\to\Rplus$ is of class-$K$   if it is   strictly increasing and $f(0)=0$. It is of class-$K_\infty$   if $f$ is of class-$K$ and $f(s)\to\infty$ as $s\to\infty$. A continuous function $\beta:\Rplus^2\to\Rplus$ is of class-$KL$  if  $\beta(\cdot,t)$ is of class-$K$ for each $t\in\Rplus$, and $\beta(s,\cdot)$ is strictly decreasing to zero for each $s\in\Rplus$. With $h:\R^n\to\R$  a $\cC^1$ function in the arguments $x_1,\dots,x_n$,  and $f:\R^n\to\R$, for each $i\in\{1,\dots,n\}$ we denote by $L_f^{(x_i)}h$ the map 
\[
x\mapsto L_{f(x)}^{(x_i)}h(x):=  \dfrac{\partial h(x)}{\partial x_i}   f(x).
\]
If $x:\R\to\R^n$ is a locally bounded function, we let $|x|_{[0,t)} := \sup_{s\in[0,t)}|x(s)|$. In this paper we consider differential equations of the form
\[
\Sigma:\; \dot x=f(x),\qquad x\in\R^n.
\]
With $X\subset\R^n$ and $\tau\in\Rplus$, we define the $\tau$-\textit{reachable set} from $X$ as
\begin{equation}\label{s:xx}
\cR_{\Sigma}^\tau(X) := \Big\{ \xi\in\R^n\suchthat\xi= x(t),\; x(0)\in X, t\ge\tau \Big\}.
\end{equation}
Clearly, $\tau_1\ge\tau_2$ implies $\cR_{\Sigma}^{\tau_2}(X)\subseteq\cR_{\Sigma}^{\tau_1}(X)$, so that the (possibly empty) $\Omega$-\textit{limit set}
\begin{equation*}
\Omega_\Sigma(X) := \lim_{\tau\to\infty}\cR_{\Sigma}^\tau(X) = \bigcap_{\tau\ge 0}\overline{\cR_{\Sigma}^\tau(X)}
\end{equation*}
is well-defined. The following theorem~\cite{Byrnes2003}, which follows directly by the definition of $\Omega_\Sigma$ and the group property of the solutions of \eqref{s:xx}, summarizes the main properties of $\Omega_\Sigma(X)$.
\begin{theorem}\label{thm:limit}
	$\Omega_\Sigma(X)$ exists and is closed. If there exists $\tau\ge 0$ such that $\cR_{\Sigma}^\tau(X)$ is bounded, then $\Omega_\Sigma^\tau(X)$ is compact, non empty, invariant and uniformly attractive from $X$. If in addition $\Omega_\Sigma(X)\subset X^\circ$, then $\Omega_\Sigma(X)$ is stable, and hence it is asymptotically stable.
\end{theorem}

\section{The Framework}
 We consider   nonlinear systems of the form
\begin{equation}\label{s:wx}
\begin{array}{lcl}
\dot w&=&s(w)\\
\dot x &=& f(w,x) + b(w,x)u \\
y&=& \begin{pmatrix}
e  \\ 
y_a\end{pmatrix} =\begin{pmatrix}
h_e(w,x)\\h_a(w,x) 
\end{pmatrix}   
\end{array}
\end{equation}
where  $x\in\R^\dx$ is the state of the plant, $y\in\R^\dy$ is the measured output, $u\in\R^\du$ is the control input, with $\du\ge\dy$,  and $w\in\R^\dw$ is an \textit{exogenous} signal modeling disturbances acting on the plant and references to be tracked (as customary in the literature of output regulation we refer to the subsystem $w$ as the \textit{exosystem}). The output $y$ is subdivided into two components, $e\in\R^\de$ and  $y_a\in\R^{\dya}$  ($\dy=\de+\dya$). The outputs $e$ are called the \textit{regulation errors}, and they represent the quantities that we aim to drive to zero asymptotically. The outputs $y_a$ represent instead additional measurable outputs that need not to vanish at the steady state, but that are available for stabilization purposes. More precisely,  we consider the problem of \textit{approximate output regulation} for system \eqref{s:wx}, which consists in finding an output-feedback regulator  ensuring boundedness of the closed-loop trajectories and that 
\[
\limsup_{t\to\infty}|e(t)|\le \varepsilon,
\]
with $\varepsilon\ge 0$ possibly a small number. In particular, we say that the problem is solved \textit{asymptotically} if $\varepsilon=0$, or   \textit{practically} if  $\varepsilon>0$ can be made arbitrarily small by opportunely tuning the regulator.

In the following we construct a regulator dealing with the problem of output regulation for \eqref{s:wx} under a number of assumption detailed hereafter. We first  assume that there exists a compact set $W\subset\R^\dw$ which is forward invariant for the exosystem, and we limit our analysis to the set of $w$ originating (and thus staying) in $W$. We also assume that the functions $s$, $f$ and $b$ are locally Lipschitz, and sufficiently smooth so that  there exist $r>0$, a set of integers $p_1,\dots,p_{r}>0$  satisfying
 $p_1+\dots+p_{r}=\dy$, a set of integers $N_1,\dots, N_{r}>0$ satisfying $p_1 N_1+\dots+p_{r}N_{r}=:N\le \dx$, and, for $i=1,\dots,r$, a set of $\R^{p_i}$-valued smooth functions\footnote{With slight abuse of notation, in the following we will call with the same symbols $\xi$ and $\zeta$ both the functions $\xi(w,x)$ and $\zeta(w,x)$ and the functions $t\mapsto \xi(w(t),x(t))$ and $t\mapsto\zeta(w(t),x(t))$.} $\{\xi^i_1(w,x),\dots,\xi^i_{N_i-1}(w,x),\zeta^i(w,x)\}$ with linearly independent differentials, such that, by letting $\xi:=\col(\xi^1,\dots,\xi^{r})\in\R^{N-\dy}$, $\xi^i:=\col(\xi^i_1,\dots,\xi^i_{N_i-1})\in\R^{p_i(N_i-1)}$, and  $\zeta:=\col(\zeta^1,\dots,\zeta^{r})\in\R^{\dy}$, we have that
\begin{equation*}%\label{e:pd_xi_g}
L_{b(w,x)}^{(x)}\xi(w,x) =0
\end{equation*}
for all $(w,x)\in\R^\dw\x\R^\dx$ and that, along the solutions to \eqref{s:wx}, $\xi$ and $\zeta$ satisfy
\begin{subequations}
	\begin{align}
	\dot{\xi} &= F \xi + H\zeta \label{s:xi}\\
	\dot{ \zeta} &= q(w,x) + B(w,x)u\label{s:zeta}\\
	y &=  T C\xi , \label{s:y_xizeta}
	\end{align}
\end{subequations}
for some locally Lipschitz functions $q:\R^\dw\x\R^\dx\to\R^\dy$ and $B:\R^\dw\x\R^\dx\to\R^{\dy\x\du}$ with $q(0,0)=0$, and in which $C:=\diag(C_1,\dots,C_{r})$ with $C_i :=\begin{pmatrix}
I_{p_i} & 0_{p_i\x p_i(N_i-2)} 
\end{pmatrix}$, 
$T\in\R^{\dy\x \dy}$ is a known permutation matrix,  and $F\in\R^{(N-\dy)\x(N-\dy)}$ and $H\in\R^{(N-\dy)\x \dy}$ are  block lower-triangular matrices whose diagonal blocks are given, respectively, by
\begin{align*}
F_{ii} &:= \begin{pmatrix}
0_{p_i(N_i-2)\x p_i} & I_{p_i(N_i-2)}\\
0_{p_i} & 0_{p_i\x p_i(N_i-2)}
\end{pmatrix}, &  H_{ii}  &:= \begin{pmatrix}
0_{ p_i (N_i-2)\x p_i}\\I_{p_i}
\end{pmatrix}.
\end{align*}
According to the partition $y=\col(e,y_a)$, we let $C_e\in\R^{\de\x (N-\dy)}$ and $C_a\in\R^{\dya\x (N-\dy)}$  be such that
\begin{equation*}
TC =\begin{pmatrix}
C_e\\C_a
\end{pmatrix}.
\end{equation*}
For simplicity, we develop here the case in which $T=I_{\dy}$, i.e., we assume that
\begin{align*}
e &= C_e\xi = \col(\xi_1^i\suchthat i=1,\dots, r_e)\\
y_a &= C_a\xi = \col(\xi_1^i\suchthat i=r_e+1,\dots, r )
\end{align*}
where $r_e$ is such that $p_1+\dots+p_{r_e}=\de$ and $r_a:=r-r_e$ (we also let $N_e=p_1N_1+\dots+p_{r_e}N_{r_e}$ and $N_a=N-N_e$).  We note though, that the result can be easily extended to arbitrary transformations $T$ by means of  more involved technical treatise. In the following we also let
\begin{equation}\label{d:xie}
\begin{aligned}
 \xi^e&:=\col(\xi^i\suchthat i=1,\dots,r_e),& \xi^a&:=\col(\xi^i\suchthat i=r_e+1,\dots,r),\\  \zeta^e&:=\col(\zeta^i\suchthat i=1,\dots,r_e), &  \zeta^a&:=\col(\zeta^i\suchthat i=r_e+1,\dots,r).
\end{aligned}
\end{equation}

\begin{remark}
	The class of systems considered includes multivariable normal forms and partial normal forms~\cite{Isidori1995book2}, with the latter that always exist locally whenever (possibly after a preliminary feedback) the system \eqref{s:wx} is \textit{a)} strongly invertible in the sense of Hirschorn and Singh~\cite{Hirschorn1979,Singh1981}, and \textit{b)} input-output linearizable\footnote{That is~\cite{Isidori1995book}, there exists a state-feedback control of the form $u=\alpha(w,x)+G(w,x)v$, with $v\in\R^{\du}$ an auxiliary input and $G$ full rank, such that the resulting system has linear input-output behavior from $v$ to $y$.}. The dimension $\du$ of the input is not constrained to be equal to the dimension $\de$ of the regulation errors nor the dimension $\dy$ of the overall measured outputs, i.e., we consider possibly \textit{non-square} systems in which $\du\ge\dy$. 
\end{remark} 
On system \eqref{s:exy:plant} we then make the following strong minimum-phase and controllability assumptions.
\begin{assumption}\label{ass:ios_mf}
	There exist a class-$KL$ function $\beta$, and class-$K$ functions $\rho_0$ and $\rho_1$, with $\rho_1$ locally Lipschitz, such that  all the solution pairs $(w,x,u)$ to \eqref{s:wx}, with $u$ locally bounded,  satisfy
	\begin{equation*}
	|x(t)| \le \beta(|x(0)| ,t) + \rho_0\big(|w|_{[0,t)}\big) + \rho_1\big(|(\xi,\zeta)|_{[0,t)} \big),
	\end{equation*}
	for all $t\in\Rplus$ for which they are defined.
\end{assumption}

\begin{assumption}\label{ass:P_L}
	There exists a $\cC^1$ map $\cP:\R^\dw\x\R^\dx \to \R^{\dy\x \dy}$, constants $\und\lambda,\,\bar\lambda>0$, and a full-rank matrix $\cL\in\R^{\du\x \dy}$ such that:
	\begin{enumerate}
		\item [a.] for all $(w,x)\in W\x\R^{\dx}$ \[\und\lambda I\le \cP(w,x)\le \bar\lambda I.\]
		\item [b.] For all $(w,x)\in W\x\R^\dx$ and all $u\in\R^\du$,\[  L_{b(w,x)u}^{(x)}\cP(w,x)   = 0.\]
		\item[c.] for all $(w,z)\in W\x\R^\dx$
		\begin{equation*}
		\cL^\top  B(w,x)^\top \cP(w,x) + \cP(w,x)B(w,x)\cL \ge I.
		\end{equation*}
	\end{enumerate}
\end{assumption}
\begin{remark}
	As in the context of normal forms and partial normal forms, $\xi$ and $\zeta$ are combinations of derivatives of the output $y$. Hence, Assumption \ref{ass:ios_mf} can be seen as an uniform (in $u$) \textit{``output-input stability''} (OIS) property \cite{Liberzon2002,Liberzon2004}  of $x$, that here plays the role of a \textit{minimum-phase} assumption. Similar minimum-phase assumptions appeared in many state-of-art frameworks~\cite{Wang2015,Wang2016,Wang2017}, in which, however, the OIS property is asked with respect to a compact attractor for $(w,x)$, typically coincident with the graph of a single-valued steady-state map $\pi:\R^{\dw}\to\R^\dx$ whose existence must be assumed. We stress, moreover, that here the minimum phase is asked with respect to the \textit{whole set of outputs} (included those that do not need to vanish at the steady state) and, thus, Assumption \ref{ass:ios_mf} is milder than usual minimum-phase assumptions, and it can be possibly obtained by adding further measurements. 
%	We also note that the result presented here can be extended to the case in which Assumption \ref{ass:ios_mf} holds only locally, provided, however, that the IOS property holds only with respect to $y=C\xi$.   
\end{remark}
\begin{remark}
	Assumption \ref{ass:P_L} is a controllability condition needed here to employ the proposed high-gain  stabilization technique. We underline that, while $\cL$ will be used in the definition of the regulator, $\cP$ needs only to exist, and it needs not to be known by the designer. As shown in Section  \ref{sec:P} below, this assumption is implicated by many customary assumptions made in the context of regulation and stabilization of partial normal forms~\cite{McGregor2006,Astolfi2013,Wang2016,Wang2017,Wang2015,Wang2015uncB,Back2009}.
\end{remark}

\section{A Postprocessing Regulator}\label{sec:regulator}
\subsection{The Regulator Structure}
With $d\in\N$ an arbitrary index, the proposed regulator is a system with state $\eta\in\R^{d\de}$ whose dynamics is described by the following equations
\begin{equation}\label{s:eta}
\begin{array}{lcl}
\dot \eta &=& \Phi(\eta) + Ge\\
u &=& \cL\big( \cK_\xi \xi + \cK_\zeta \zeta+ \cK_\eta\eta_1\big) ,
\end{array}
\end{equation}
with $\Phi$ and $G$ having  the form
\begin{align*}
\Phi( \eta )&: = \begin{pmatrix} 0 & I_{\de} & 0 & \cdots  & 0\\ 0 & 0 & I_{\de} & \cdots & 0\\ 
\cdot & \cdot & \cdot & \cdots & \cdot\\
0 & 0 & 0  &\cdots & I_{\de}\\
& & \phi(\eta)
\end{pmatrix},& G &:= \begin{pmatrix} G_1 \\G_2 \\ \vdots \\G_d \end{pmatrix}
\end{align*}
in which $\phi:\R^{d \de }\to \R^{\de}$ is a  Lipschitz function  satisfying 
\begin{equation}\label{e:bound_phi}
\begin{aligned}
% |\phi(\eta)|&\le C_\phi, & \forall &\eta\in\R^\deta\\
 |\phi(\eta_1)-\phi(\eta_2)|&\le L_\phi|\eta_1-\eta_2|, & \forall&\eta_1,\,\eta_2\in\R^\deta \\
 \phi(0) &= 0
\end{aligned}
\end{equation} 
for some $L_\phi>0$, and in which $G_i\in\R^{\de\x \de}$, $\cK_\xi\in\R^{\dy\x (N-\dy)}$, $\cK_\zeta\in\R^{\dy\x \dy}$ and $\cK_\eta\in\R^{\dy\x \de}$ are control gains to be designed, with $\cK_\eta$  that has the form
\begin{equation}\label{d:cK_eta}
\cK_\eta = \begin{pmatrix}
\cK_\eta'\\
0_{\dya\x \de}
\end{pmatrix},
\end{equation}
for some $\cK'_\eta\in\R^{\de\x \de}$. All these degrees of freedom have to be fixed according to the forthcoming Proposition \ref{prop:boundedness} and Theorem \ref{thm:main}.

\begin{remark}
	We give here a {\it partial state feedback} result employing the auxiliary variables $\xi$ and $\zeta$ (i.e. combinations of derivatives of the measured output $y$). We note, however, that a purely output-feedback regulator can be easily obtained by augmenting \eqref{s:eta} with a partial-state observer~\cite{Wang2015,Atassi1999,Teel1995}.
\end{remark}

As mentioned in the introduction, the chicken-egg dilemma~\cite{Bin2018b} prevents a sequential design of the different parts of regulator, leading us to a more complicated construction. Since the ideal choice of $\phi$ in \eqref{s:eta} depends on the actual choice of the matrices $\cL$, $\cK_\xi$, $\cK_\zeta$ and $\cK_\eta$, before discussing the choice of $\phi$, we state a preliminary uniform boundedness result, which allows us to individuate a class of possible steady-state trajectories on which $\phi$ can be tuned.  

\subsection{Existence of a Steady State}
The closed-loop system reads as follows
\begin{equation}\label{s:cl}
\Sigma_{cl}:\;\left\{ \begin{array}{lcl}
\dot w&=&s(w)\\
\dot x &=& f(w,x) + b(w,x) \cL\big( \cK_\xi \xi(w,x) + \cK_\zeta \zeta(w,x)+ \cK_\eta\eta_1\big)  \\ 
\dot \eta &=& \Phi(\eta) + Gh_e(w,x) .
\end{array}\right. 
\end{equation}
The following proposition, proved in Section \ref{sec:proof_prop_boundedness}, shows that the parameters of the regulator \eqref{sec:regulator} can be chosen to ensure (semi-global) uniform boundedness of the closed-loop trajectories, and hence the existence of a steady state.  
\begin{proposition}\label{prop:boundedness}
Suppose that Assumptions \ref{ass:ios_mf} and \ref{ass:P_L} hold and consider the regulator \eqref{s:eta}, with  $d>1$ and $\phi$ satisfying \eqref{e:bound_phi} for some $L_\phi>0$ and $\cL$ given by Assumption \ref{ass:P_L}. Then there exists a class-$K$ functions $\rho_{cl}$ and, for each pair of compact sets $X\subset\R^\dx$ and $H\subset\R^\deta$, a class-$KL$ function   $\beta_{cl}$  and  matrices  $\cK_\xi$, $\cK_\zeta$, $\cK_\eta$ and $G$, with $\cK_\eta$ of the form \eqref{d:cK_eta} with $\cK_\eta'$ invertible, such that all the solutions to $\Sigma_{cl}$ originating in $W\x X\x H$ are complete and satisfy
\begin{equation}\label{claim:prop1}
|(x(t),\eta(t))|  \le \beta_{cl}(|(x(0),\eta(0))|,t)  +  \rho_{cl}(|w|_{[0,t)})  ,
\end{equation}  
for all $t\ge 0$.
\end{proposition}

The claim of Proposition \ref{prop:boundedness} shows that, despite the particular $\phi$ implemented in \eqref{s:eta}, if it satisfies \eqref{e:bound_phi}, then for each compact set of initial conditions the other degrees of freedom of the regulator can be chosen to ensure that   the closed-loop system $\Sigma_{cl}$ is uniformly bounded\footnote{That is, there exists $M>0$ such that every trajectory $(w,x,\eta)$ of the closed-loop system \eqref{s:wx}, \eqref{s:eta} originating in $W\x X\x H$ satisfies $|(w(t),x(t),\eta(t))|\le M$ for each $t\in\Rplus$.} from $W\x X\x H$ and hence, by Theorem \ref{thm:limit}, there exists a compact  attractor
\begin{equation}\label{d:cA}
\cA:= \Omega_{\Sigma_{cl}}(W\x X\x H)
\end{equation}
which is invariant and uniformly attractive from $W\x X\x H$. 
Moreover, since $\rho_{cl}$ does not depend on $X$ and $H$, the latter sets can be thought, without loss of generality, to be large enough so that $\cA\subset (W\x X\x H)^\circ$, which implies that $\cA$ is also asymptotically stable.

\subsection{Tuning the Internal Model}
If \eqref{e:bound_phi} holds,   Proposition \ref{prop:boundedness} ensures   the existence of $G$,  $\cK_\xi$, $\cK_\zeta$, $\cK_\eta$, such that there exists a compact attractor $\cA$ for the trajectories closed-loop system $\Sigma_{cl}$ originating in $W\x X\x H$. In this section, we consider  the restriction of $\Sigma_{cl}$ on $\cA$, i.e. the system\footnote{We remark that, being $\cA$ invariant, all the solutions to $\Sigma_{cl}^\cA$ are complete.}
\begin{equation}\label{s:cl_cA}
\Sigma_{cl}^\cA:\;\left\{ \begin{array}{lcl}
\dot w&=&s(w)\\
\dot x &=& f(w,x) + b(w,x) \cL\big( \cK_\xi \xi(w,x) + \cK_\zeta \zeta(w,x)+ \cK_\eta\eta_1\big)  \\ 
\dot \eta &=& \Phi(\eta) + Gh_e(w,x)  
\end{array}\right. \qquad (w,x,\eta)\in\cA,
\end{equation}
and we investigate the conditions that $\phi$ must satisfy on $\cA$ to guarantee a given asymptotic performance, expressed in terms of the bound (ideally zero) on the regulation error $e=h_e(w,x)$ along the solutions to \eqref{s:cl_cA}. This conditions are then used in the forthcoming Theorem \ref{thm:main} to complement the claim of Proposition \ref{prop:boundedness} with an asymptotic bound on the regulation error depending on ``how well'' the function $\phi$ is chosen.

Substituting the expression of $u$ into \eqref{s:zeta} yields
\begin{equation}\label{e:dzeta_u}
\dot\zeta = q(w,x) +   B(w,x)\cL ( \cK_\xi \xi + \cK_\zeta\zeta + \cK_\eta\eta_1).
\end{equation}
Let us define the projected set $\cO$ as
\begin{equation*}
 \cO:= \Big\{ (w,x)\in W\x\R^\dx \suchthat  (w,x,\eta)\in\cA \Big\} ,
\end{equation*}
let $q^e$ and $q^a$ be functions with values in $\R^{\de}$ and $\R^{\dya}$ respectively such that $q(w,x)=\col(q^e(w,x),q^a(w,x))$, and    denote for brevity
\[
D(w,x) := B(w,x)\cL \in\R^{\dy\x \dy}.
\] 
Partition $D(w,x)$, $\cK_\xi$ and $\cK_\zeta$ as:
\begin{equation}\label{d:D_K}
\begin{aligned}
D&:= \begin{pmatrix}
D^{e,e} & D^{e,a}\\D^{a,e} & D^{a,a}
\end{pmatrix}, & \cK_\xi&:= \begin{pmatrix}
\cK_\xi^{e,e} & \cK_\xi^{e,a}\\\cK_\xi^{a,e} & \cK_\xi^{a,a}
\end{pmatrix}, & \cK_\zeta&:= \begin{pmatrix}
\cK_\zeta^{e,e} & \cK_\zeta^{e,a}\\\cK_\zeta^{a,e} & \cK_\zeta^{a,a}
\end{pmatrix}
\end{aligned}
\end{equation}
for some $D^{i,j}(w,x),\cK_\xi^{i,j},\cK_\zeta^{i,j}\in\R^{n_i\x n_j}$, $i,j\in\{e,a\}$. Then, in view of \eqref{d:cK_eta}, Equation \eqref{e:dzeta_u} gives
\begin{equation}\label{e:zeta_e_dot}
\dot{\zeta}^e  = q^e(w,x) + D^{e,e}(w,x)\Big( \cK_\xi^{e,e}\xi^e + \cK_\xi^{e,a}\xi^a+\cK_\zeta^{e,e}\zeta^e+\cK_\zeta^{e,a}\zeta^a + \cK_\eta'\eta_1 \Big) +  D^{e,a}(w,x)\Big( \cK_\xi^{a,e}\xi^e + \cK_\xi^{a,a}\xi^a+\cK_\zeta^{a,e}\zeta^e+\cK_\zeta^{a,a}\zeta^a \Big).
\end{equation}

In view of \eqref{e:zeta_e_dot}, the regulation error $e$, along with its derivatives, vanish on the attractor $\cA$ if and only if $\xi^e=0$, $\zeta^e=0$ (both implying $\dot\xi^e=0$),  and
\begin{equation}\label{e:zeta_e_dot2}
0  = q^e(w,x) + D^{e,e}(w,x)\Big(  \cK_\xi^{e,a}\xi^a +\cK_\zeta^{e,a}\zeta^a + \cK_\eta'\eta_1 \Big) +  D^{e,a}(w,x)\Big(  \cK_\xi^{a,a}\xi^a +\cK_\zeta^{a,a}\zeta^a \Big) 
\end{equation}
for all $(w,x)\in\cO$. By solving \eqref{e:zeta_e_dot2} for $\eta_1$, we thus obtain an ideal steady-state value for $\eta_1$, given by the  function $\eta^\star_1:\cO\to\R^{\de}$ satisfying
\begin{equation}\label{e:eq_eta1_star} 
D^{e,e}(w,x)\cK_\eta'\eta^\star_1(w,x)  = - q^e(w,x) - D^{e,e}(w,x)\Big( \cK^{e,a}_\xi \xi^a(w,x) + \cK^{e,a}_\zeta\zeta^a (w,x)\Big)  - D^{e,a}(w,x)\Big( \cK^{a,a}_\xi \xi^a(w,x) + \cK^{a,a}_\zeta\zeta^a  (w,x)\Big) 
\end{equation}
for all $(w,x)\in\cO$. Since $\eta_1$ is the first component of the internal model unit $\eta$ of \eqref{s:eta}, and since in the ideal condition in which $e=0$ the system $\eta$ coincides with the autonomous equation
\begin{equation}\label{s:eta_autonomous}
\dot\eta = \Phi(\eta),
\end{equation}
then for each solution to $\Sigma_{cl}^\cA$, \eqref{e:eq_eta1_star} individuates an ideal steady-state value for $\eta$ given by
\begin{equation}\label{d:etasr}
\eta\sr(t) := \col\begin{pmatrix}
\eta_1\sr(t), &
\dot\eta_1\sr(t),& 
\cdots,&
\eta_1\sr{}^{(d-1)}(t)
\end{pmatrix}
\end{equation}
in which we denoted for brevity $\eta\sr_1(t)=\eta\sr_1(w(t),x(t))$. 
In view of the structure of the map $\Phi$ in \eqref{s:eta}, such an ideal steady state trajectory $\eta\sr$ may be a solution of \eqref{s:eta_autonomous} only if $\eta_1\sr(t)$ satisfies
\begin{equation}\label{IMp_t}
\eta_1\sr{}^{(d)}(t) = \phi\Big( \eta_1\sr(t),\ \dot\eta_1\sr(t),\ \dots, \eta_1\sr{}^{(d-1)}(t)\Big)
\end{equation}
almost everywhere.
%in which $\eta\sr_{d+1}$ is defined as in \eqref{e:eta_star} and is well-defined under Assumption \ref{ass:etastar}. The fact that \eqref{IMp_t} must hold for each trajectory originating in $\cA$, implies that the function $\phi$ of \eqref{s:eta} must be chosen so that
%\begin{equation}\label{e:IMp}
%\phi(\eta\sr(w,x)) = \eta\sr_{d+1}(w,x), \qquad \forall (w,x)\in\cO.	
%\end{equation}
Condition \eqref{IMp_t} expresses indeed  the \textit{internal model property}, that is, the property of the regulator to generate all the ideal  steady-state control actions $\eta_1\sr(t)$ needed to keep the regulation error $e$ to zero. We also observe that the chicken-egg dilemma is strongly present in \eqref{IMp_t}, since $\eta_1\sr$ and its derivatives, and thus the correct value of $\phi$ to be implemented, depend on the closed-loop trajectories and thus, in particular, on stabilization gains  $\cK_\xi$, $\cK_\zeta,$ and $\cK_\eta$, the latter   dependent from the Lipschitz constant $L_\phi$ of $\phi$. 
%Nevertheless, it is the a priori knowledge that the designer has on the asymptotic trajectories of the closed-loop systems that guides the choice of $\phi$, ideally satisfying \eqref{IMp_t}. 
\begin{remark}
	We   observe that in the linear case the chicken-egg dilemma is broken by the fact that, no matter how the stabilization gains are chosen, if a compact attractor $\cA$ exists then all the closed-loop trajectories  have the same modes of the exosystem (the closed-loop system \eqref{s:cl}, indeed, is a linear stable system driven by the exosystem). Therefore, a function $\phi$ satisfying \eqref{e:bound_phi} and \eqref{IMp_t} can be fixed a priori according to the knowledge of the exosystem dynamics. A similar situation also takes place in a nonlinear setting if $\dot w=0$. If $d$ and $\phi$ are   taken as $d=1$ and $\phi(\eta)=0$, and if the closed-loop system $\Sigma_{cl}$ can be made ``contractive'' via stabilization, then the steady-state trajectories (i.e., the solutions to $\Sigma_{cl}^\cA$) are constant, and thus \eqref{IMp_t} holds. This is in fact the well-known integral action~\cite{Astolfi2017}.
\end{remark}

% the restriction on $\cO$ of the map $\eta\sr$ defined as
%\begin{equation}\label{e:eta_star}
%\begin{aligned}
%\eta\sr(w,x) &= \col\Big( 
%\eta\sr_1(w,x),\ 
%\eta\sr_2(w,x),\ 
%\dots,\ 
%\eta\sr_d(w,x)\Big)\\
%\eta\sr_i(w,x) &= \pd{\eta\sr_{i-1}(w,x)}{w}s(w) + \pd{\eta\sr_{i-1}(w,x)}{x} \big( f(w,x) + b(w,x)u \sr(w,x)\big),\qquad i=2,\dots,d
%\end{aligned}
%\end{equation} 
%in which $\eta\sr_1$ is a solution to \eqref{e:eq_eta1_star} and $u\sr$, defined as
%\begin{equation*}
%u\sr(w,x) := \cL\left( \cK_\xi\begin{pmatrix}
%0\\\xi^a(w,x)
%\end{pmatrix}+  \cK_\zeta\begin{pmatrix}
%0\\\zeta^a(w,x)
%\end{pmatrix} + \cK_\eta\eta_1\sr(w,x) \right),
%\end{equation*}
%coincides with the ideal control action keeping $e$ to zero in view of \eqref{e:eq_eta1_star}. 
%

The existence of  functions $\eta_1\sr(w,x)$ solving \eqref{e:eq_eta1_star}  in $\cO$ and such that $\eta_1\sr(t)$ is $d$-times differentiable, and thus the existence of the ideal steady state $\eta\sr$ for $\eta$, follows by the  assumptions below (see also the forthcoming Remark \ref{rmk:eta1}).
\begin{assumption}\label{ass:etastar}
	The functions $s$, $f$ and $b$ are $\cC^d$.
\end{assumption}
\begin{assumption}\label{ass:Dee}
	There exists a map $\cM:\R^\dw\x\R^\dx \to \R^{\de\x \de}$, $\cC^1$ on an open set including $\cO$,   such that:
	\begin{enumerate}
		\item [a.] $\cM(w,x)>0$ in $\cO$.
		\item [b.] For all $(w,x)\in \cO$ and all $u\in\R^\du$,\[  L_{b(w,x)u}^{(x)}\cM(w,x)   = 0.\]
		\item[c.] For all $(w,x)\in \cO$
		\begin{equation*}
		D^{e,e}(w,x)^\top \cM(w,x) + \cM(w,x)D^{e,e}(w,x) \ge I.
		\end{equation*}
	\end{enumerate}
\end{assumption}
\begin{remark}
	Assumption \ref{ass:Dee} is a controllability condition resembling those of Assumption \ref{ass:P_L}, but limited to the submatrix $D^{e,e}(w,x)$ inside the attractor $\cO$. 
	In addition to Assumption \ref{ass:P_L}, this condition ensures that the regulation error contracts to zero when the internal model reaches its ideal steady state.	While in the case in which $y=e$ this condition is directly implicated by Assumption \ref{ass:P_L}, in general it is not and it has to be assumed. We stress that, as it is the case for $\cP$ in Assumption \ref{ass:P_L}, the matrix $\cM$ is not used in the construction of the control law, and it is thus not required to be known by the designer. In the forthcoming Section \ref{sec:P}, we complement this assumption by showing that it holds, together with Assumption \ref{ass:P_L}, in many state-of-art relevant frameworks of high-gain regulation and stabilization of multivariable systems.
\end{remark}
\begin{remark}\label{rmk:eta1}
	 Assumption \ref{ass:Dee} implies that $D^{e,e}(w,x)$ is nonsingular in $\cO$. This, together with Assumption \ref{ass:etastar} and the fact that $\cK_\eta'$ can be taken invertible, implies that \eqref{e:eq_eta1_star} admits a unique solution in $\cO$ such that $\eta_1\sr(t)$ is $C^d$ for each solution $(w,x,\eta)$ of $\Sigma_{cl}^\cA$.
\end{remark}

\subsection{Asymptotic Performances}
Apart from the chicken-egg dilemma, fulfilling condition \eqref{IMp_t} in general requires a very precise knowledge of the function $\eta\sr_1$ and its derivatives, which strongly depend on the plant and exosystem dynamics. Therefore, uncertainties in the model of the plant and exosystem strongly reflect into uncertainties in the right internal model function to implement, potentially ruining the asymptotic performance of the regulator~\cite{Bin2018CDC}. 
This motivates introducing, for each solution $(w,x,\eta)$ to $\Sigma_{cl}^\cA$, the following quantity 
\begin{equation*}%\label{1n:d:delta}
\delta(t) := \phi\big(\eta\sr(t)\big) - \eta\sr_1{}^{(d)}(t),
\end{equation*}
representing the  \textit{internal model mismatch} along $(w,x,\eta)$, i.e. the   error that the system $\eta$ of \eqref{s:eta} attains  in modeling the process that generates the ideal control action $\eta\sr_1(t)$. We also define the worst-case mismatch as the quantity
\begin{equation}\label{d:bardelta}
\bar \delta := \sup\left\{ |\delta|_\infty\suchthat (w,x,\eta)\text{ solution to }\Sigma_{cl}^\cA  \right\},
\end{equation}
which, we stress, in general depends also on the control gains $\cK_\xi$, $\cK_\zeta$ and $\cK_\eta$.

The following theorem, which represents the main result of the paper,  characterizes the   asymptotic properties of the proposed regulator \eqref{s:eta} and, in particular, relates the worst-case internal model mismatch \eqref{d:bardelta} to the asymptotic bound on the regulation error.

\begin{theorem}\label{thm:main}
Suppose that Assumptions \ref{ass:ios_mf}, \ref{ass:P_L}, \ref{ass:etastar} and \ref{ass:Dee} hold, and consider the regulator \eqref{s:eta} with $d>1$ and $\phi$ satisfying \eqref{e:bound_phi} for some $L_\phi>0$, and with $\cL$ given by Assumption \ref{ass:P_L}. Then,  for each pair of compact sets $X\subset\R^\dx$ and $H\subset\R^\deta$ of initial conditions and each $\varepsilon>0$, there exist  matrices $\cK_\xi,\,\cK_\zeta,\,\cK_\eta$ and $G$, and a compact set $\cA$ of the form \eqref{d:cA}, such that every solution to the closed-loop system $\Sigma_{cl}$  originating in $W\x X\x H$ is uniformly attracted by $\cA$  and  satisfies\begin{equation}\label{e:bound_e}
\limsup_{t \to \infty}|e(t)|\le \varepsilon \bar\delta 
\end{equation}
uniformly in the initial conditions.
\end{theorem}

%According to the proof of Proposition \ref{prop:postproc}, which is postponed to Section \ref{sec:proof}, the degrees of freedom that characterize the regulator \eqref{s:eta} are chosen by following a ``high-gain'' strategy. 
%The only parameters that remain truly arbitrary are the order $d$ of the internal model and the bound $C_\phi$ on $\phi$. Lower values of $C_\phi$ yield lower values of the high-gain parameters, but reduce the ``representation capabilities'' of the internal model unit, thus potentially increasing $\delta(w,x)$. 
%The functions $\upsilon$ and $\upsilon_{d+1}$, on the other hand, are obtained by the recursion \eqref{e:upsilon1}, \eqref{e:upsilon} and, hence, they are strongly dependent on the plant's and the exosystem's functions $s$, $f$ and $g$, and  on the control parameters $\cK_\xi$, $\cK_\zeta$ and $\cK_\eta$. Therefore, their perfect knowledge cannot be assumed while fixing $C_\phi$, and this dependence between the internal model unit and the stabilizer's parameters is a manifestation of the \textit{``chicken-egg'' dilemma}\cite{Bin2018b}.
 
Theorem \ref{thm:main}, proved in Section \ref{sec:pf_thm}, claims that the stabilization parameters in \eqref{s:eta} can be chosen to ensure the existence of a compact attractor for the closed-loop system \eqref{s:cl}, and that the asymptotic bound on the regulation error is proportional to the worst-case internal model mismatch on the attractor. The claim of Theorem \ref{thm:main} is thus an \emph{approximate} regulation result. Furthermore, if the worst-case internal model mismatch \eqref{d:bardelta} does not depend on the control parameters $\cL$, $\cK_\xi$, $\cK_\zeta$, $\cK_\eta$ and $G$, then the  asymptotic bound on the error  can be reduced arbitrarily by opportunely choosing the control parameters   to lower the proportionality constant $\varepsilon$ in \eqref{e:bound_e}. This, in turn, makes the claim of the theorem a  \textit{practical} output regulation result. 
Since $X$ and $H$ are arbitrary, moreover, the result is \textit{semiglobal} in the plant's and internal model's initial conditions. We remark that, as evident in \eqref{e:eq_eta1_star}, in the canonical ``square'' setting in which $b(w,x)=0$,  $\dya=0$ (i.e. $y=e$) and $\du=\de$, the map $\eta_1\sr(w,x)$, its derivatives, and hence $\bar\delta$, can be always bounded uniformly in the control parameters, thus making the result of Theorem \ref{thm:limit} always practical.

Moreover, we observe that if the worst-case internal model mismatch \eqref{d:bardelta} is zero, i.e. the internal model includes a copy of the process generating the ideal steady state $\eta\sr$ for \eqref{s:eta}, then \emph{asymptotic regulation} is achieved, that is,
\begin{equation*}%\label{e:ar}
\lim_{t\to\infty} e(t) = 0.
\end{equation*}
holds uniformly.
The intertwining between internal model and the stabilizer and the possible presence of model's uncertainties make    $\bar\delta=0$ very difficult to satisfy, thus making the conditions of asymptotic regulation de facto not constructive.  Nevertheless, Theorem \ref{thm:main} individuates a clear sufficient condition for asymptotic regulation, given by equation \eqref{IMp_t}, which expresses in this setting the nonlinear version of the \textit{internal model property}: the regulator must embed a copy of the process that generates the ideal steady state control action (i.e. $\eta_1^\star$ above) that makes the set in which the error vanishes invariant. 

Finally, we observe that this chicken-egg dilemma is way more evident when nonvanishing outputs are used for stabilization, as they need to be compensated at the steady state by the output of the internal model (see \eqref{e:eq_eta1_star}). In this respect, we also note that, as it is the case of the linear regulator, the feedback of auxiliary outputs might also have a simplifying effect on the internal model.

\section{Example} 
In this section, we present an example showing how the postprocessing paradigm presented in the previous sections can be used to approach a problem for which no solution of the preprocessing type is known. We consider the system
\begin{equation}\label{s:exy:plant}
\begin{array}{lcl}
\dot w &=& s(w)\\
\dot x_1 &=& f_1(x_1) + \gamma_1(w,x_2) + x_3 \\
\dot x_2 &=& \gamma_2(w,x) +  u_1 + u_2 \\
\dot x_3 &=& \gamma_3(w,x)  -b(w) u_1+(1-b(w))u_2
\end{array}
\end{equation}
with regulation error
\[
e:=x_2,
\]
in which all the functions $f_1$, $\gamma_i$ and $b$ are smooth, and where $w$ ranges in a compact invariant set $W\subset\R^\dw$. We observe that if  the function  $\gamma_2$   satisfies $\partial \gamma_2(0)/ \partial x_i = 0$ for $i=1,3$, the linear approximation of \eqref{s:exy:plant} at $0$ is not detectable from $e$. Thus, $e$ is not enough to stabilize \eqref{s:exy:plant}, and additional outputs are needed. We specifically assume to have available for feedback the other two variables, i.e. the additional output
\[
y_{a} := \col(x_1,x_3).
\]
We observe that $y_a$ does not necessarily   vanish at the ideal steady state in which $e=0$. We also observe that a control strategy based on a preliminary ``pre-stabilizing'' inner-loop employing $y_{a}$ to reduce to a canonical case in which the pre-stabilized plant is stabilizable from $e$   is hard to imagine. In fact the inputs $u_1$ and $u_2$ affect both the equations of $\dot x_2$ and $\dot x_3$, and $x_3$ has no relative degree with respect to any of the two inputs, since both $b(w)$ and $1-b(w)$ may vanish. Therefore,  both $u_1$ and $u_2$ must be used in case $x_3$ is pre-stabilized, thus leaving no further degree of freedom to deal with $e$. 
As a consequence, this case does not fit into any of the previous preprocessing frameworks in which only $e$ can be used for feedback.

In the rest of the section we build a regulator of the form \eqref{s:eta} dealing with this case. We suppose to know a function $\kappa$   such that\footnote{We observe that if $f_1$ is locally Lipschitz, then $\kappa$ can be taken linear over compact sets.}
\begin{equation}\label{e:exy:kappa}
s\big( f_1(s) -\kappa(s)\big)\le 0, \qquad\forall s \in\R 
\end{equation}	
and we define the variable $\zeta:=\col(\zeta_1,\zeta_2)$ as
\begin{align*}
\zeta_1&:=x_2,&\zeta_2 &:= x_3+x_1+\kappa(x_1),
\end{align*}
which satisfy
\begin{equation}\label{s:exy:zeta} 
\dot\zeta  = q(w,x) + B(w,x) u 
\end{equation}
with 
\begin{equation*}
\begin{aligned} 
q(w,x)&:= \begin{pmatrix}
\gamma_2 (w,x )\\
\gamma_3(w,x)+\big(1+\kappa'(x_1)\big)\big(f_1(x_1)+\gamma_1(w,x_2)+x_3\big)
\end{pmatrix},& 
B(w,x)&:=\begin{pmatrix}
1&1\\ -b(w) & 1-b(w)
\end{pmatrix}.
\end{aligned}
\end{equation*} 
Hence, \eqref{s:exy:plant}-\eqref{s:exy:zeta} has the form \eqref{s:wx}, \eqref{s:zeta}. We further notice that $|(x_2,x_3)|\le  |\zeta|+|x_1+\kappa(x_1)|$, with $x_1$ that fulfills
\begin{equation*}
\dot x_1 = -x_1 + f_1(x_1)-\kappa(x_1) + \gamma_1(w,\zeta_1)+ \zeta_2.
\end{equation*}
In view of \eqref{e:exy:kappa}, simple computations show that the $x_1$ subsystem is input-to-state stable relative to the origin and with respect to the input $\gamma_1(w,\zeta_1)+\zeta_2$, and hence Assumption \ref{ass:ios_mf} is fulfilled.

With $\alpha>0$ satisfying $4\alpha> 1+\sup_{w\in W}(1+b(w))^2$, let 
\begin{align*}
\cL_0&:= \begin{pmatrix}
\alpha^2&0\\0&\alpha
\end{pmatrix}, & \cP(w,x)&:=\begin{pmatrix}
1+b(w)^2&b(w)\\b(w)&1
\end{pmatrix}.
\end{align*} 
Then the high-frequency matrix $B(w,x)$	fulfills
\begin{equation*}%\label{e:exy:BLS} 
\cL^\top_0B(w,x)^\top\cP(w,x)+\cP(w,x)B(w,x)\cL_0   =\begin{pmatrix}
2\alpha^2&\alpha(1+b(w))\\ \alpha(1+b(w))&2\alpha 
\end{pmatrix}=: M(w). 
\end{equation*}
As $2\alpha^2>0$ and 
\begin{equation*}
\det M(w)  = 4\alpha^3-\alpha^2(1+b(w))^2 = \alpha^2(4\alpha-(1+b(w))^2) >4\alpha^2 + \alpha^2\big( 4\sup_{w\in W}(1+b(w))^2-(1+b(w))^2\big) >4\alpha^2,
\end{equation*} 
$M(w)$ is positive definite and there exists $m>0$ such that, for all $x\in\R^2$, 
\begin{equation}\label{e:exy:Mm}
x^\top M(w)x\ge  m|x|^2.
\end{equation}
Therefore the pair $(\cL,\cP)$, with $\cL:=\cL_0/m$, satisfies Assumption \ref{ass:P_L} for each compact subset   $X\subset\R^\dx$.

Moreover, Assumption \ref{ass:etastar} holds by construction, and Assumption \ref{ass:Dee} holds with $\cM=1$, since the quantity $D(w,x)=B(w,x)\cL$ defined as in \eqref{d:D_K} satisfies
\begin{equation*}
D^{e,e}(w,x) = \alpha^2.
\end{equation*}   

According to Section \ref{sec:regulator}, and by following the proof of Proposition \ref{prop:boundedness} and Theorem \ref{thm:main}, the regulator has the form \eqref{s:eta}, with the stabilizing action given by
\begin{equation*}%\label{ex:u}
u =    -\ell \dfrac{\alpha}{m}\begin{pmatrix}
\alpha (e-\eta_1)\\
x_3+x_1+\kappa(x_1) 
\end{pmatrix}
\end{equation*}
in which $\ell>0$ is a design parameter to be taken sufficiently large and $\eta_1$ is the first component of the internal model unit
\begin{equation*}
\begin{array}{lcl}
\dot\eta_i &=& \eta_{i+1} + g^i h_i e.\quad i=1,\dots,d-1\\
\dot \eta_d &=& \phi(\eta) + g^dh_d e
\end{array}
\end{equation*}
in which the parameters $h_i$ are the coefficients of a Hurwitz polynomial, $g>0$ is a high-gain parameter, and the dimension $d$ and the $\phi$ are degrees of freedom to be fixed according to the available knowledge of $s(w)$ and of the ideal steady state of the internal model given by \eqref{d:etasr}. According to Proposition \ref{prop:boundedness}, for all sufficiently large $\ell$ and $g$ the closed-loop system has a compact attractor, on which $d$ and $\phi$ must be tuned.
In particular,   $\eta_1\sr$ is obtained as in \eqref{e:eq_eta1_star} and by imposing in \eqref{s:exy:plant} $e=0$ and $\eta=\eta\sr$, thus obtaining
\begin{equation}\label{ex:etasr}
\begin{aligned}
\eta_1\sr &= -\dfrac{m}{\ell\alpha^2} \gamma_2(w,(x_1,0,x_3)) + \dfrac{1}{\alpha}(x_3+x_1+\kappa(x_1))\\
\dot x_1 &= f_1(x_1)+\gamma_1(w,0)+x_3\\
\dot x_3 &= \gamma_3(w,(x_1,0,x_3)) -b(w) \gamma_2(w,(x_1,0,x_3)).
\end{aligned}
\end{equation}
In the following simulation, we set
\begin{align*}
\dw&=2, & s(w)&=\begin{pmatrix}
0 & 1\\-1 & 0
\end{pmatrix}w, & W&=\ball_3, &
f_1(x_1)&=0, & \kappa(x_1)&=0, \\ \gamma_1(w,x_2)&=0, &\gamma_2(w,x) &= q \exp(2w_1^2),&  \gamma_3(w,x)&=w_1^2, & b(w)&=w_1, & \alpha&= 5
\end{align*}
with $q\ge0$. Then \eqref{e:exy:Mm} holds with $m=1$, and Equations \eqref{ex:etasr} reduce to
 \begin{equation*}%\label{ex:etasr2}
 \begin{aligned}
 \eta_1\sr &= -\dfrac{1}{\ell\alpha^2} q\exp(w_1^2) + \dfrac{1}{\alpha}(x_3+x_1),&
 \dot x_1 &= x_3,&
 \dot x_3 &= w_1^2 -q w_1\exp(2w_1^2).
 \end{aligned}
 \end{equation*}
We first consider the case in which $q=0$. It can be shown that, in this case, $\eta_1\sr$ can be generated by a linear system of dimension $d=5$ and of the form \eqref{s:eta_autonomous} in which
\begin{equation*}
\phi(\eta) = -4\eta_4.
\end{equation*}
Figure \ref{fig:1} shows a simulation in this first case in which $q=0$, $d=5$ and $\phi$ is chosen as above. The system originates at $w(0)=(1,0)$, $x(0)=(3,5,-2)$ and $\eta(0) = 0$,  and   the  control gains are chosen as $g=5$ and $\ell=5$.

In the second case in which $q$ is set to $1$, $\eta_1\sr$ cannot be generated by a  linear system of dimension $5$ and thus asymptotic regulation is not achieved. In this setting, however, the internal model mismatch \eqref{d:bardelta} can be bounded uniformly in the control parameters, and practical regulation is achieve.
Figure \ref{fig:2} thus shows   a simulation obtained with the same regulator in the case in which $q=1$ and $g=5,8,10$, by showing that the asymptotic error can be reduced by increasing $g$. 

\begin{figure}[t]
	\centering
	\begin{subfigure}{.5\linewidth}
		\centering
		\includegraphics[width=\linewidth]{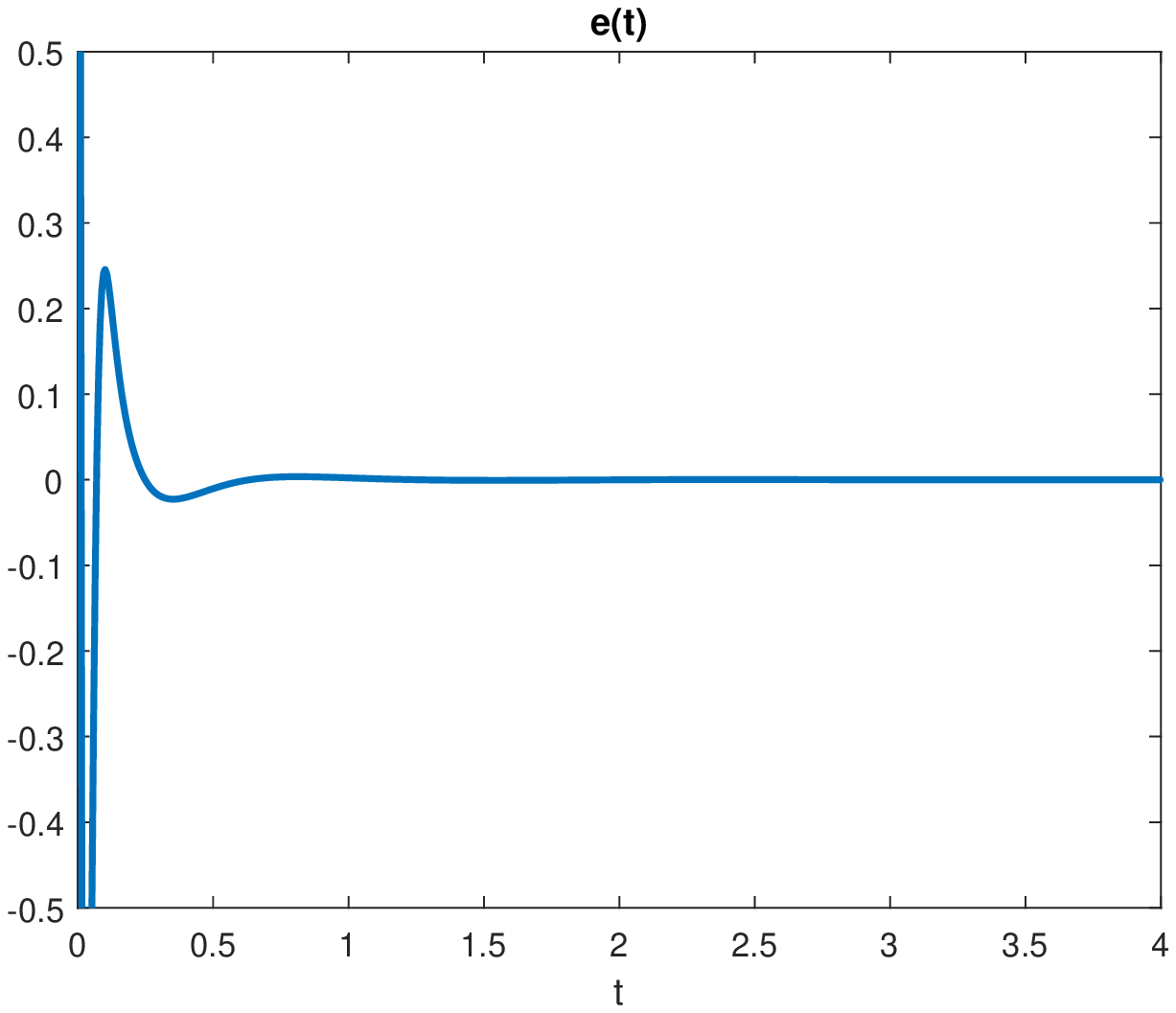}
		\label{fig:f1}
	\end{subfigure}%
	\begin{subfigure}{.5\linewidth}
		\centering
		\includegraphics[width=\linewidth]{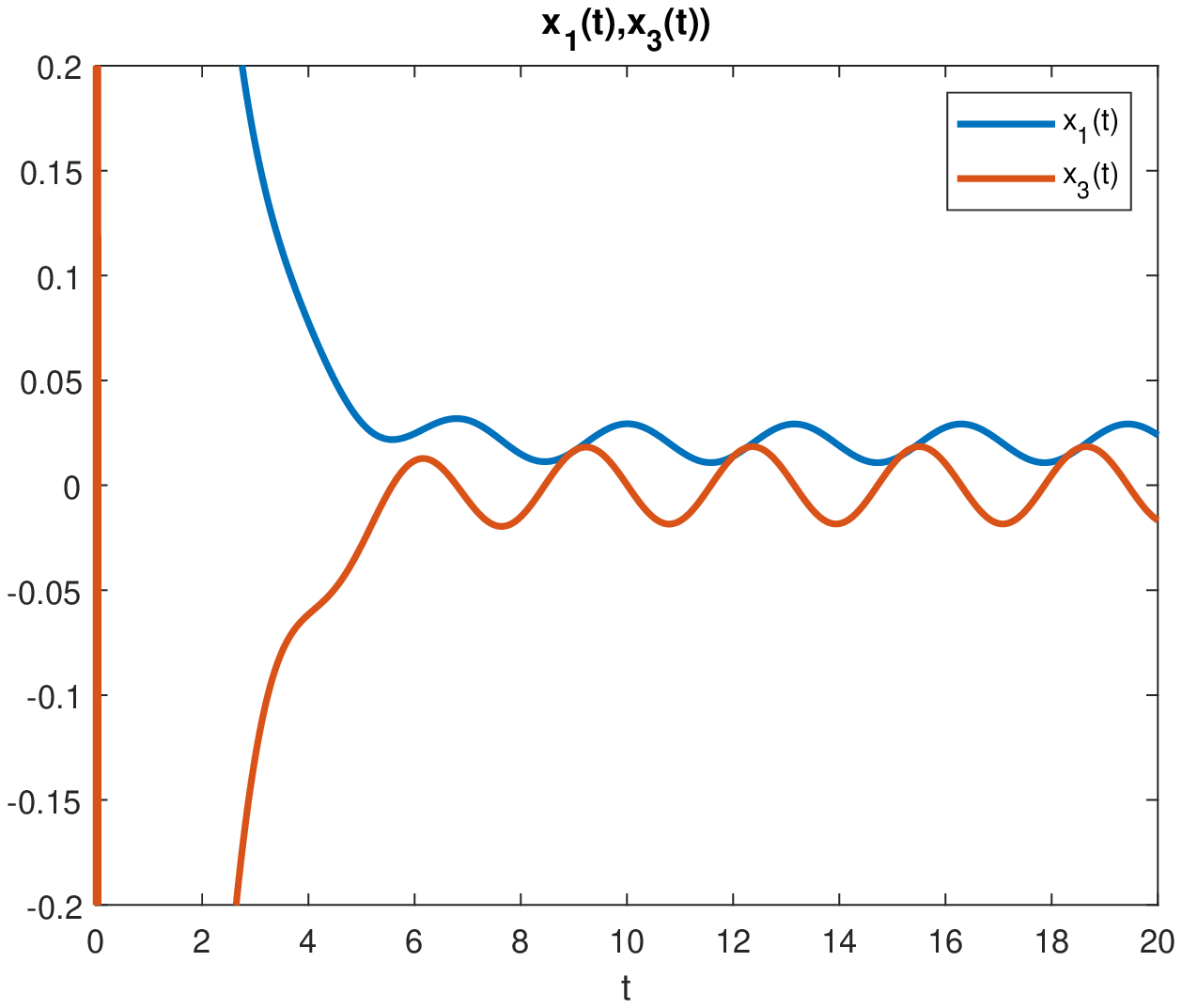} 
		\label{fig:f2}
	\end{subfigure}
	\caption{The left plot shows the time evolution of the regulation error $e(t)=x_2(t)$ which is asymptotically vanishing. The right plot shows the time evolution of the other two measurements $y_a(t)=(x_1(t),x_3(t))$ that have a non-zero steady state.}
	\label{fig:1}
\end{figure}

\begin{figure}[t]
	\centering  
		\includegraphics[width=\textwidth,clip,trim=4.8em 0 3.6em 0]{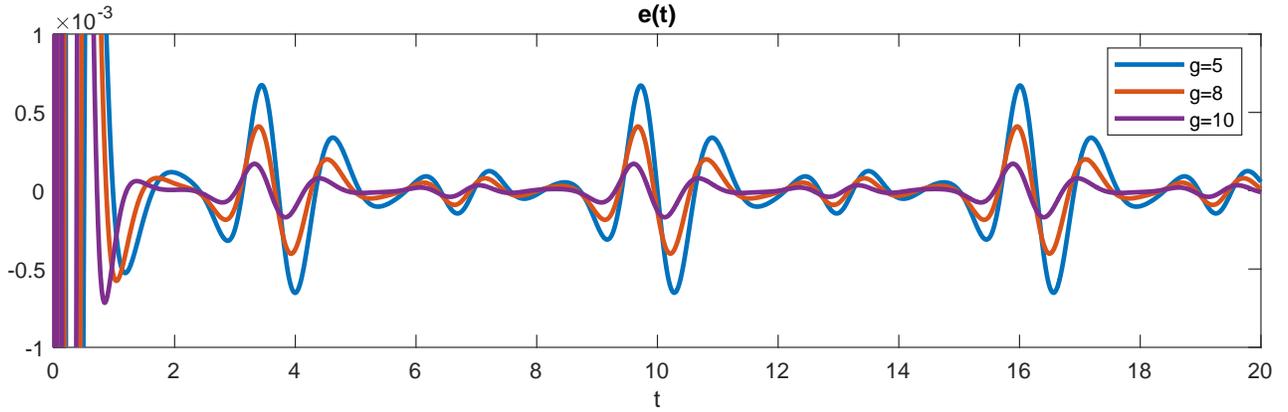} 
	\caption{Time evolution of the regulation error $e(t)=x_2(t)$ for $g=5,8,10$.}
	\label{fig:2}
\end{figure}

\section{Proofs}

\subsection{Proof of Proposition \ref{prop:boundedness}}\label{sec:proof_prop_boundedness}

Partition $\eta$ as $\eta=\col(\eta_1,\dots,\eta_d)$ and, for each $\eta_i\in\R^{\de}$,   partition $\eta_i$ as $\eta_i=\col(\eta_i^1,\dots,\eta_i^{r_e})$ with $\eta_i^j\in\R^{p_j}$. Consider the   change of coordinates: 
\begin{equation}\label{d:chi_components}
\begin{aligned}
\forall i=1,\dots, r_e :&\quad \left\{\begin{aligned}
\xi^i_1 &\mapsto  \chi^i_1 := \xi^i_1 +  \eta_1^i,\\ \xi^i_j&\mapsto  \chi^i_j := \xi^i_j, \qquad j=2,\dots, N_i-1
\end{aligned}\right.\\
\forall i=r_e,\dots, r :&\qquad \xi^i\mapsto  \chi^i:=\xi^i.\\
&\qquad e\mapsto  \bar e := e+\eta_1.
\end{aligned}
\end{equation}
By letting $\chi:=\col(\chi^1,\dots,\chi^{r})$, with $\chi^i:=\col(\chi^i_1,\dots,\chi^i_{N_i-1})$, the change of variables \eqref{d:chi_components} can be compactly rewritten as
\begin{equation*}%\label{d:chi}
\begin{array}{lcl}
\chi &:=& \xi + C_e^\top  \eta_1\\
\bar e &=& C_e\chi . 
\end{array}
\end{equation*}
From \eqref{s:eta}, and since by construction $FC_e^\top =0$, we obtain 
\begin{equation}\label{s:chi}
\dot\chi  = (F+C_e^\top G_1C_e)\chi + H\zeta + C_e^\top (\eta_2 - G_1\eta_1).
\end{equation} 
We state now the following lemma, proved in the Appendix.
\begin{lemma}\label{lem:kappa}
	For any  $\epsilon>0$, there exists $K\in\R^{\dy\x (N-\dy)}$, class-$KL$ functions $\beta_\chi, \beta_{\bar e}$ and $a_1>0$, such that the system \eqref{s:chi} with output $\bar e$ and with input $\zeta=\bar\zeta + K \chi$, being $\bar\zeta\in\R^\dy$ an auxiliary input, satisfies
	\begin{align}
	&|\chi^i(t)|\le \beta_\chi( |\chi(0)|,t/\epsilon) + a_1 \Big( |\bar \zeta_{[1,i]}|_{[0,t)}  +   |G_1||\eta_{1}|_{[0,t)}+|\eta_2|_{[0,t)}\Big)\label{e:lemma_k_claim_iss}\\
	&|\bar e(t)| \le    \beta_{\bar e}(| \chi(0)|,t/\epsilon) + \epsilon |\bar \zeta^e|_{[0,t)} +   \epsilon\Big(|G_1||\eta_1|_{[0,t)}+|\eta_2|_{[0,t)}\Big) \label{e:lemma_k_claim_oiss}
	\end{align} 
	for all $i=1,\dots,r$   and with  $\bar\zeta^e:=\bar\zeta_{[1,r_e]}$.
\end{lemma} 
Pick $(h_1,\dots, h_d)\in\R^d$ such that the polynomial $\lambda^d+h_1\lambda^{d-1}+\dots+h_{d-1}\lambda+h_d$ has $d$ roots with negative real part and, with $g>0$ a control parameter, let
\begin{equation} \label{d:Gi}
 G_i := g^i h_i I_{n_e}. 
\end{equation}
Let $\Delta(g):=\diag(1,g,\dots,g^{d-1})$ and change variables as
\begin{equation*}
\eta  \mapsto \mu  := \Delta(g)\inv\eta.
\end{equation*}
In the new variables we obtain
\begin{equation*}
\dot \mu = \Delta(g)\inv\Big( A\Delta(g)\mu + E\phi(\Delta(g)\mu) +  G( \bar e-\Gamma\mu_1)\Big)
\end{equation*}
with $A\in\R^{\deta\x\deta}$,  $E\in\R^{\deta\x\de}$ and $\Gamma\in\R^{\de\deta}$ defined as
\begin{align*}
A&:=\begin{pmatrix}
0_{\de} & I_{\de} & 0_{\de} & \cdots & 0_{\de}\\
0_{\de} &0_{\de}& I_{\de} &     \cdots & 0_{\de}\\
\vdots & & & \ddots & \vdots\\
0_{\de}  &   & \cdots &   & I_{\de}\\
0_{\de}  &   & \cdots &   & 0_{\de}\\
\end{pmatrix}, & E&:=\begin{pmatrix}
0_{\de}\\0_{\de}\\\vdots\\0_{\de}\\I_{\de}
\end{pmatrix}, & \Gamma&:=\begin{pmatrix}
I_\de & 0_\de &\cdots & 0_\de 
\end{pmatrix}.
\end{align*}
Noting that: 
\begin{align*}
\Delta(g)\inv A\Delta(g) &=gA, &
\Delta(g)\inv E &= g^{1-d} E, &
\Delta(g)\inv G &= g R,\ \text{with}\  R:= \col\big( h_i I_{n_e}\suchthat i=1,\dots, d\big), 
\end{align*}
then, by letting $M:=A-R\Gamma$, we obtain
\begin{equation}\label{e:d_mu}
\dot\mu = gM\mu + g^{1-d} E \phi(\Delta(g)\mu) + gRC_e\chi .
\end{equation}
Let $L_\phi>0$ such that \eqref{e:bound_phi} holds. Then 
\begin{equation}\label{bound:muL}
|g^{1-d} E \phi(\Delta(g)\mu)|\le L_\phi |\mu|
\end{equation}
holds for all $\mu\in\R^{d\de}$.
 By definition of $(h_1,\dots,h_d)$, $M$ is Hurwitz. Then, in view of \eqref{bound:muL}, standard high-gain arguments~\cite{Khalil2013} can be used to show that there exists $g\sr>0$, a class-$KL$ function $\beta_{\mu}$ and $a_2>0$, such that, for all $g>g\sr$, the system $\mu$ fulfills
\begin{equation*} 
|\mu(t)|\le \beta_{\mu}(|\mu(0)|,gt)  + a_2 |\bar e|_{[0,t)} .
\end{equation*} 
for all $t\ge 0$. Thus, by letting $\beta_{\eta} = g^{d-1}\beta_{\mu}$, using the fact that $|\eta_i|\le g^{i-1}|\mu_i|$, we also obtain that the system $\eta$ fulfills
\begin{equation}\label{e:bound_eta}
\begin{aligned}
|\eta_i(t)|&\le \beta_{\eta}(|\eta(0)|,t)   + a_2g^{i-1} |\bar e|_{[0,t)} 
%\\|\eta_1(t)|&\le \beta_{\eta}(|\eta(0)|,t) + \dfrac{a_2}{g^d} C_\phi + a_2  |\bar e|_{[0,t)}
\end{aligned}
\end{equation}
for all $i=1,\dots,d$.
With $\epsilon$ being any small number so that
\[
\epsilon < \epsilon^\star_1(g):=\dfrac{1}{a_2(|G_1|+g)},
\]
being $a_2$ the same as in \eqref{e:bound_eta}, let $K=K(g)$ be the corresponding matrix produced by Lemma \ref{lem:kappa}, and change variables as
\begin{equation}\label{1n:d:barzeta}
\zeta\mapsto \bar\zeta:= \zeta - K \chi .
\end{equation}
Then, the bounds \eqref{e:lemma_k_claim_iss}-\eqref{e:lemma_k_claim_oiss} hold with $|\eta_2|_{[0,t)}\le g |\mu|_{[0,t)}$ and, since $\epsilon a_2 (g+|G_1|)<1$, standard small-gain arguments~\cite{Jiang1994} imply the existence of a class-$KL$ function $\beta_{\chi,\eta}$ and $a_3>0$ such that the following bound holds
\begin{equation}\label{bound:chi_eta}
|(\chi(t),\eta(t))|\le \beta_{\chi\eta}(|(\chi(0),\eta(0))|,gt) + a_3|\bar \zeta|_{[0,t)}.
\end{equation}
Moreover, since for any class-$K$ function $\rho$, $\rho(a+b)\le \rho(2a)+\rho(2b)$, Assumption \ref{ass:ios_mf} yields
\begin{equation}\label{bound:x}
\begin{aligned}
|x(t)| &\le \beta(|x(0)|,t) + \rho_0(|w_{[0,t)}|) + \rho_1( |(\xi,\zeta)|_{[0,t)})\\
&\le  \beta(|x(0)|,t) + \rho_0(|w_{[0,t)}|) + \rho_1\big( (1+|K|) |\chi|_{[0,t)} + |\eta_1|_{[0,t)}  + |\bar\zeta|_{[0,t)}  \big)\\
&\le \beta(|x(0)|,t)+\rho_0(|w_{[0,t)}|) + \rho_1'\big(|(\chi,\eta)|_{[0,t)}\big)+\rho_1''\big(|\bar\zeta|_{[0,t)}\big)
\end{aligned}
\end{equation}
with $\rho_1'(s):=\rho_1(2(2+|K|) s)$ and $\rho_1''(s):=\rho_1(2s)$. Thus, we conclude by    \eqref{bound:chi_eta} and \eqref{bound:x} that there exist $\beta_{x\chi\eta}$ of class-$KL$ and $\rho_3$ of class-$K$ such that every solution of $\Sigma_{cl}$ satisfies
\begin{equation}\label{bound:x_chi_eta}
|(x(t),\chi(t),\eta(t))|\le \beta_{x\chi\eta}(|(x(0),\chi(0),\eta(0))|,t) + \rho_0(|w|_{[0,t)}) + \rho_3(|\bar\zeta|\sn)
\end{equation}
for al $t\ge 0$. 

In view of \eqref{s:zeta} and \eqref{s:chi}, $\bar\zeta$ fulfills
\begin{equation}\label{s:bar_zeta}
\dot{\bar \zeta} =   \varrho (\eta,\chi,\bar\zeta) + q(w,x)+ B(w,x) u
\end{equation} 
where 
\begin{equation}\label{d:varrho}
\varrho (\eta,\chi,\bar\zeta)  := -K\Big((F+C_e^\top G_1C_e)\chi + H\zeta + C_e^\top (\eta_2 - G_1\eta_1)\Big).
\end{equation}
With $\ell>0$ a design parameter to be fixed, let in \eqref{s:eta}  
\begin{align}\label{d:u}
\cK_\xi&:=\ell K,&\cK_\zeta&:=-\ell I_{n_e}, &\cK_\eta&:=\ell K C_e^\top \,.
\end{align}
The fact that $\cK_\eta$ can be taken of the form \eqref{d:cK_eta}, with $\cK_\eta'$ invertible follows by the presence of $C_e^\top$ in \eqref{d:u}, and by the construction of $K$ in Lemma \ref{lem:kappa}.
In the new coordinates we then have
\begin{equation*}
u  =  \ell \cL\big(  K(\xi+ C_e^\top \eta_1) -\zeta     \big)   = -\ell\cL  \bar\zeta  
\end{equation*}
and thus \eqref{s:bar_zeta} yields
\begin{equation}\label{s:dot_bar_zeta}
\dot{\bar \zeta} = \varrho (\eta,\chi ,\bar\zeta)+q(w,x) -\ell B(w,x)\cL\bar\zeta .
\end{equation} 
We fix $\ell$ on the basis of the following Lemma, whose proof is in the Appendix.
\begin{lemma}\label{lem:zeta}
		Consider an equation of the form \eqref{s:dot_bar_zeta}, with $\varrho$ and $q$ locally Lipschitz. Under Assumption \ref{ass:P_L} there exist $\beta_{\bar\zeta}$ of class-$KL$ and, for each compact set $\bar X\subset\R^{\dx}$,   constants $a_4>0,\,\ell^\star_0(g,K)>0$ such that, for all $\ell>\ell^\star_0(g,K)$  the following holds
		\begin{equation}\label{e:bound_bar_zeta}
		|\bar\zeta(t)|\le \beta_{\bar\zeta}(|\bar\zeta(0)|,\ell t) + \dfrac{a_4}{\ell}\left( |(w,x,\chi,\eta)|_{[0,t)}\right),
		\end{equation}
		for all $t\in\Rplus$ such that $(w(s),x(s))\in W\x \bar X$ for all $s\in[0,t]$.
\end{lemma} 

Pick a compact set of initial conditions for $(x,\eta)$ of the form $X\x H$ and let $\pi_1,\pi_2>0$ be such that $(w,x,\eta)\in W\x X\x H$ implies
\begin{align*}
|\chi| &= |\xi(w,x)+C_e^\top \eta| \le \pi_1\\
|\bar \zeta| &= |\zeta(w,x)-K(\xi(w,x)+C_e^\top \eta)|\le \pi_2.
\end{align*}   
Pick $\varpi_1>0$ so that
\begin{equation}\label{e:varpi1}
\varpi_1 \ge 1+ \max\left\{ \max_{(x,\eta)\in X\x H} |(x,\eta)|+ \pi_1 ,\  \sup_{(x,\eta)\in X\x  H} \beta_{x\chi\eta}( |(x,\eta)|+\pi_1,0) + \sup_{w\in W}\rho_0(|w|) + \rho_3\big(\beta_{\bar \zeta}(\pi_2,0)+1\big)\right\},
\end{equation}
{which can be fixed at this stage since, according to Lemma \ref{lem:zeta}, $\beta_{\bar \zeta}$ does not depend on $\bar X$.} 
With $\varpi_2>\varpi_1$  define
\begin{equation}\label{d:barX}
\bar X := \overline{\ball_{\varpi_2}}.
\end{equation}
let $\ell_0\sr(g,K)$ and $a_4$ be given by Lemma \ref{lem:zeta} for $\bar X$ chosen as in \eqref{d:barX}, and pick
\begin{equation*}
\ell>\ell\sr_1(g,K):=\max\left\{\ell\sr_0(g,K), a_4\left( \max_{w\in W}|w|+ \varpi_2 \right)\right\}.
\end{equation*}
Then, for any solution to $\Sigma_{cl}$ originating in $W\x X\x H$  it holds that
\begin{equation}\label{bound:barzeta_2}
|\bar\zeta(t)|\le \beta_{\bar \zeta}(\pi_2,0) + 1.
\end{equation}
for all $t\ge 0$ such that  $|(x,\chi,\eta)|\sn < \varpi_2$, as the latter condition implies $x(s)\in\bar X$ for all $s\in[0,t]$.

Now, suppose that a solution to $\Sigma_{cl}$ exists that originates in $W\x X\x H$ and for some $\bar t>0$ satisfies $|(x(\bar t),\chi(\bar t),\eta(\bar t))|>\varpi_2$. By definition of $\varpi_1$, $|x(0),\chi(0),\eta(0)|< \varpi_1<\varpi_2$, so that by continuity there exists $\bar s\in(0,\bar t)$ such that $|x(\bar s),\chi(\bar s),\eta(\bar s))|=\varpi_1$ and $|(x,\chi,\eta)|_{[0,\bar s)}\le \varpi_1$. However, by \eqref{bound:x_chi_eta}, \eqref{e:varpi1} and \eqref{bound:barzeta_2}, this implies that
\begin{equation*}
|\bar\zeta|_{[0,\bar s)} \le \beta_{\bar \zeta}(\pi_2,0)+1
\end{equation*}
and hence
\begin{align*} 
\varpi_1 &= |(x(\bar s),\chi(\bar s),\eta(\bar s))| \\
&\le \beta_{x\chi\eta}\big(|x(0),\chi(0),\eta(0)|,0\big) + \rho_0(|w|_{[0,\bar s)}) +  \rho_3(|\bar\zeta|_{[0,\bar s)}) \\
%&\le \sup_{(x,\eta)\in X\x H} \beta_{x\chi\eta}( |(x,\eta)|+\pi_1,0) + \sup_{w\in W}\rho_0(|w|) + \rho_3\big(|\bar\zeta|_{[0,\bar s)}\big)\\
&\le \sup_{(x,\eta)\in X\x H} \beta_{x\chi\eta}( |(x,\eta)|+\pi_1,0) + \sup_{w\in W}\rho_0(|w|) + \rho_3\big(\beta_{\bar \zeta}(\pi_2,0)+1\big)\\
&\le \varpi_1-1
\end{align*}
which is a contradiction. Thus, we claim that every solution to $\Sigma_{cl}$ originating in $W\x X\x H$ is complete and satisfies $(x(t),\chi(t),\eta(t))|\le \varpi_2$ and $|\bar\zeta(t)|\le \beta_{\bar \zeta}(\pi_2,0)+1$ for all $t\ge 0$.
As $\rho_3$ is locally Lipschitz, then this implies that there exists $a_5>0$ such that every trajectory originating in $W\x X\x H$ satisfies
\begin{equation*}%\label{bound:x_chi_eta_2}
|(x(t),\chi(t),\eta(t))|\le \beta_{x\chi\eta}(|(x(0),\chi(0),\eta(0))|,t) + \rho_0(|w|_{[0,t)}) + a_5 |\bar\zeta|\sn
\end{equation*}
for all $t\ge 0$.
In view of \eqref{e:bound_bar_zeta}, which holds for every $t\ge0$, standard small-gain arguments show that there exists $\beta_{cl}$ of class-$KL$ such that, with  $\rho_{cl}(\cdot)=\rho_0(\cdot)+1$ and for all
\begin{equation*}
\ell > \ell\sr(g,K) := \max\{ \ell_1\sr(g,K),\, a_4 a_5  \},
\end{equation*}
the claim \eqref{claim:prop1} of the proposition holds. \endproof

%%%%%%%%%%%%%%%%%%%%%%%%%%%%%%%%%

\subsection{Proof of Theorem \ref{thm:main}}\label{sec:pf_thm}

By Proposition \ref{prop:boundedness}, there exist $g\sr>0$, $K=K(g)\in\R^{p\x p(N-p)}$ and $\ell\sr(g,K)>0$ such that, for all $g>g\sr$ and $\ell>\ell\sr(K,g)$, the choices \eqref{d:Gi} and \eqref{d:u} guarantee that the claim \eqref{claim:prop1} of the proposition holds, and in particular that there exists a compact invariant attractor $\cA$ defined in \eqref{d:cA}.    
Pick now a solution $(w,x,\eta)$ to $\Sigma_{cl}^\cA$ and, define the signal $\eta\sr$ according to \eqref{d:etasr}, which under Assumptions \ref{ass:etastar} and \ref{ass:Dee} is well defined in view of Remark \ref{rmk:eta1}.  With reference to the coordinates \eqref{d:chi_components} and \eqref{1n:d:barzeta}, consider the  change of variables 
\begin{align*}
\bar e &\mapsto \tilde e := \bar e- \eta\sr_1\\
\eta &\mapsto \tilde \eta := \eta - \eta\sr\\
\chi &\mapsto \tilde\chi := \chi - C_e^\top \eta\sr_1 \\
\bar \zeta &\mapsto \tilde\zeta := \bar\zeta + KC_e^\top \eta\sr_1 = \zeta-K\tilde\chi.
\end{align*}
Then 
\begin{equation*}
\tilde e  = C_e\chi - \eta\sr_1 = C_e\tilde\chi ,
\end{equation*}
and, in view of \eqref{s:chi} and noting that  $FC_e^\top =0$ and $C_eC_e^\top =I_{n_e}$, we also have
\begin{equation}\label{s:tildechi}
\dot {\tilde \chi} = (F+C_e^\top G_1C_e+HK)\tilde \chi + H\tilde\zeta + C_e^\top (\tilde\eta_2-G_1\tilde\eta_1).
\end{equation}
With $\Delta(g):=\diag(I,gI,\dots, g^{d-1} I)$, let
\begin{equation*}
\tilde\mu:=\Delta(g)\inv\tilde\eta, 
\end{equation*}
the same arguments used in the proof of Proposition \ref{prop:boundedness} in dealing with $\mu$ can be used to show that $\tilde\mu$ satisfies
\begin{equation*}
\dot{\tilde\mu} = gM\tilde \mu + g^{1-d}E\big( \phi(\Delta(g)\tilde\mu+\eta\sr)-\eta\sr_{1}{}^{(d)}\big) + gRC_e\tilde\chi 
\end{equation*} 
with $M$, $E$ and $R$ defined as in \eqref{e:d_mu}.

Then, as $\phi$ satisfies \eqref{e:bound_phi}, for some $L_\phi>0$, it holds that
\begin{align*}
|\phi(\Delta(g) \tilde\mu+\eta\sr)-\eta\sr_1{}^{(d)}| &\le |\phi(\Delta(g)\tilde\mu+\eta\sr)-\phi(\eta\sr)|+|\phi(\eta\sr)-\eta\sr_1{}^{(d)}|\\
%& \le \min\{ 2 C_\phi, L_\phi g^{d-1}|\tilde\mu|\} + |\delta(w,x)|\\
&\le  L_\phi g^{d-1}|\tilde\mu|  + \bar\delta.
\end{align*}
where $\bar\delta$ is defined  in \eqref{d:bardelta}.

Standard high-gain arguments  can be thus used to prove that, for large enough $g$, the following estimate holds
\begin{equation}\label{bound_tmu} 
|\tilde \mu(t)|\le \beta_{\tilde\mu}(|\tilde\mu(0)|,t) + \dfrac{q_1}{g^d}   \bar\delta   + q_1 |\tilde e|_{[0,t)} .
\end{equation}
for some class-$KL$ function $\beta_{\tilde \mu}$ and some $q_2>0$. We can thus assume without loss of generality that $g$ is taken large enough to guarantee this. 

Since in view of \eqref{d:Gi}, $|G_1||\tilde\eta_1|+|\tilde\eta_2|\le q_2 g |\tilde\mu|$, for some $q_2>0$ independent from $g$, then
%-----------------------------------------------
%Then, by letting $\beta_{\tilde \eta} = g^{d-1}\beta_{\tilde\mu}$, and by using the fact that $|\tilde\eta_i|\le g^{i-1}|\tilde\mu_i|$ we also obtain 
%\begin{equation}\label{e:bound_teta}
%\begin{aligned}
%|\tilde\eta_i(t)|&\le \beta_{\tilde\eta}(|\tilde\eta(0)|,t) +   \dfrac{q_2}{g^{d+1-i}}  \sup_{(w,x)\in\cO}|\delta(w,x)|  + q_2g^{i-1} |\tilde e|_{[0,t)]} 
%%\\|\tilde\eta_1(t)|&\le \beta_{\tilde\eta}(|\tilde\eta(0)|,t) +   \dfrac{q_2}{g^d}   \sup_{(w,x)\in\cO}|\delta(w,x)| + q_2  |\tilde e|_{[0,t)}.
%\end{aligned}
%\end{equation}
%
in view of  Lemma \ref{lem:kappa} and \eqref{s:tildechi}, we obtain that $K$ can be chosen without loss of generality  so that there exist $\beta_{\tilde \chi}$ and $\beta_{\tilde e}$ of class-$KL$, $q_3>0$ and $\epsilon < (2q_1  g)\inv$ for which the following bounds hold
\begin{equation}\label{e:bound_tchi}
\begin{aligned}
&|\tilde\chi^e(t)|\le \beta_\chi( |\tilde\chi(0)|,t) + q_3\Big( |\tilde \zeta^e|_{[0,t)} +  g|\tilde\mu|\sn\Big)\\
&|\tilde e(t)| \le  \tilde \beta_{e}(| \tilde\chi(0)|,t) + \epsilon\Big( |\tilde \zeta^e|_{[0,t)}+ g|\tilde\mu|\sn\Big)  ,
\end{aligned} 
\end{equation}
in which $\tilde\chi^e$ and $\tilde\zeta^e$ are defined as in \eqref{d:xie}.
Therefore, by small-gain arguments~\cite{Jiang1994} \eqref{bound_tmu} and \eqref{e:bound_tchi} imply the following bound 
\begin{equation}\label{e:bound_tmu_tchi}
%|(\tilde\mu(t),\tilde \chi^e(t))|\le \beta_{\tilde\eta\tilde \chi^e}(|(\tilde\eta(0),\tilde \chi (0))|,t) + q_4 |\tilde\zeta^e|_{[0,t)}   + q_5 g^{-d} \sup_{(w,x)\in\cO}|\delta(w,x)|,
\limsup_{t \to \infty}|(\tilde\mu(t),\tilde \chi^e(t))|\le  q_4\limsup_{t \to \infty}  |\tilde\zeta^e(t)|    + q_4 g^{1-d}\bar\delta,
\end{equation}
for some  $q_4>0$.

In view of \eqref{d:u}, and since
\begin{equation*}%\label{e:xizeta_tilde}
\begin{aligned}
\xi&=\chi-C_e^\top \eta_1 = \tilde\chi -C_e^\top \tilde\eta_1 , &
\zeta &= \bar\zeta+K\chi = \tilde\zeta +K\tilde\chi ,
\end{aligned}
\end{equation*}
we obtain
%\begin{align*}
%u &= \cL\Big(\cK_\xi\xi+\cK_\zeta\zeta+\cK_\eta\eta_1 \Big) =\cL\Big( \cK_\xi(\tilde\chi-C_e^\top \tilde\eta_1)+\cK_\zeta(\tilde\zeta+K \tilde\chi )+\cK_\eta(\tilde\eta_1+\eta\sr_1(w,x))   \Big) \\
%&=\cL\Big( \cK_\xi(\tilde\chi-C_e^\top \tilde\eta_1) +\cK_\zeta(\tilde\zeta+K\tilde\chi)+\cK_\eta\tilde\eta_1 \Big)  + \cL\Big(\cK_\xi\lambda_\xi(w,x)+\cK_\zeta\lambda_\zeta(w,x)+\cK_\eta\upsilon_1(w,x)\Big)\\
%&=\ell\cL\Big( K(\tilde\chi-C_e^\top \tilde\eta_1) -I_{n_e}(\tilde\zeta+K\tilde\chi)+KC_e^\top \tilde\eta_1 \Big)  + \cL\Big(\cK_\xi\lambda_\xi(w,x)+\cK_\zeta\lambda_\zeta(w,x)+\cK_\eta\upsilon_1(w,x)\Big)\\
%&= \ell\cL\tilde\zeta + \cL\Big(\cK_\xi\lambda_\xi(w,x)+\cK_\zeta\lambda_\zeta(w,x)+\cK_\eta\upsilon_1(w,x)\Big). 
%\end{align*}
\begin{align*}
u &= \cL\Big(\cK_\xi\xi+\cK_\zeta\zeta+\cK_\eta\eta_1 \Big) =\cL\Big( \cK_\xi(\tilde\chi-C_e^\top \tilde\eta_1)+\cK_\zeta(\tilde\zeta+K \tilde\chi )+\cK_\eta(\tilde\eta_1+\eta\sr_1(w,x))   \Big)  = -\ell\cL\tilde\zeta + \cL\cK_\eta\eta\sr_1(w,x).
\end{align*}
In view of the definition of $\eta_1\sr$ in \eqref{e:eq_eta1_star}, and of the block-diagonal structure of $K$ in Lemma \ref{lem:kappa}, by letting
\begin{equation*}
Q:=\begin{pmatrix}
I_{\de} \\ 0_{ (\dy-\de)\x \de}
\end{pmatrix},
\end{equation*} we thus obtain
\begin{align*}
\dot{\tilde\zeta}^e &=  \varrho^e(\tilde\chi^e,\tilde\zeta^e,\tilde\eta) + q^e(w,x) + Q^\top B(w,x)\Big(-\ell \cL \tilde\zeta + \cL\cK_\eta\eta\sr_1(w,x)\Big) = \varrho^e(\tilde\chi^e,\tilde\zeta^e,\tilde\eta)  -\ell Q^\top B(w,x)\cL Q \tilde\zeta^e \\&= \varrho^e(\tilde\chi^e,\tilde\zeta^e,\tilde\eta)- \ell D^{e,e}(w,x) \tilde\zeta^e
\end{align*}
with $\varrho^e$ a linear map that, in view of the lower-triangular structure of $F$ and of the diagonal structure of $H$, only depends on $(\tilde\chi^e,\tilde\zeta^e,\tilde\eta)$ and satisfies $\varrho^e(\tilde\chi^e,\tilde\zeta^e,\tilde\eta)=Q^\top \varrho(\tilde\eta,\tilde\chi,\tilde\zeta)$, with $\varrho$   defined in \eqref{d:varrho}.
 We now observe that, since $\rho_{cl}$ in the claim of Proposition \ref{prop:boundedness} does not depend on $\ell$, then   in view of Assumption \eqref{ass:Dee} and on the fact that $(w(t),x(t))\in\cO$ implies the existence of $\und\lambda',\bar\lambda'>0$ such that $\und\lambda' I \le \cM(w,x)\le \bar\lambda' I$,   arguments similar to those of Lemma \ref{lem:zeta} applied  with $q=0$ and $\bar X=\{x\in\R^\dx\suchthat (w,x)\in\cO\}$ to the Lyapunov candidate
\[
V_e(w,x):=\sqrt{\tilde\zeta^e{}^\top \cM(w,x)\tilde\zeta^e} ,
\]
 can be used to show that for large enough $\ell$, $\tilde\zeta^e$ satisfies
\begin{equation*}
%|\tilde\zeta^e(t)|\le \beta_{\tilde\zeta^e}(|\tilde\zeta^e(0)|,t) + \dfrac{q_6}{\ell} |(\tilde\chi^e,\tilde\eta)|_{[0,t)}
\limsup_{t \to \infty}|\tilde\zeta^e(t)|\le  \dfrac{q_5 g}{\ell} \limsup_{t \to \infty} |(\tilde\chi^e(t),\tilde\mu(t))|
\end{equation*}
for some $q_5>0$ independent from $\ell$ and $g$. Hence, in view of \eqref{e:bound_tmu_tchi}, and since $\ell$ can be taken without loss of generality larger enough so that
\begin{equation}\label{e:ell}
\ell> 2 q_5 q_4 g, 
\end{equation}
then we obtain
\begin{align*}
\limsup_{t \to \infty}|(\tilde \mu(t),\tilde\chi^e(t))|&\le \dfrac{q_4q_5 g}\ell \limsup_{t \to \infty}|(\tilde \mu(t),\tilde\chi^e(t))| + q_4 g^{1-d} \bar\delta\\
\limsup_{t \to \infty}|\tilde\zeta^e(t)|&\le \dfrac{q_4q_5g}\ell \limsup_{t \to \infty}|\tilde\zeta^e(t)| + \dfrac{q_4q_5g}\ell g^{1-d} \bar\delta.
\end{align*}
As \eqref{e:ell} implies both 
\begin{align*}
\dfrac 1 2 &< 1-\dfrac{q_4q_5 g}{\ell}, &  \dfrac{q_4q_5 g}{\ell} &< \dfrac{1}{2},
\end{align*}
then we obtain
\begin{equation}\label{e:limsups}
\begin{aligned}
\limsup_{t \to \infty}|(\tilde \mu(t),\tilde\chi^e(t))|&\le 2\left(1-\dfrac{q_4q_5 g}\ell \right) \limsup_{t \to \infty}|(\tilde \mu(t),\tilde\chi^e(t))| \le   2q_4 g^{1-d}  \bar\delta\\
\limsup_{t \to \infty}|\tilde\zeta^e(t)|&\le 2\left(1- \dfrac{q_4q_5g}\ell \right)\limsup_{t \to \infty}|\tilde\zeta^e(t)| \le  g^{1-d} \bar\delta.
\end{aligned} 
\end{equation}

Now,  pick a point $(\bar w,\bar x,\bar \eta)$ in $\cA$.
Define arbitrarily a sequence $\{t_n\}_{n\in\N}$ of positive scalars $t_n$ satisfying $t_n\to \infty$.
 As $\cA$ is invariant, it is backwards invariant. Hence, for each $n\in\N$, we can find a solution $(w^n,x^n,\eta^n)$ to $\Sigma_{cl}^\cA$ which satisfies
\begin{equation*}
(w^n(t_n),x^n(t_n),\eta^n(t_n))=(\bar w,\bar x,\bar \eta).
\end{equation*}  
In view of \eqref{e:limsups}, which holds for any solution of $\Sigma_{cl}^\cA$  and thus in particular for each $(w^n,x^n,\eta^n)$, it follows that for each $\varepsilon>0$, there exists $N(\varepsilon)\in\N$ such that the quantities $\tilde\mu^n:=\Delta(g)\inv ( \eta^n- \eta\sr(w^n,x^n)  )$ and $\tilde\chi^e{}^n:=\chi(w^n,x^n)-C_e^\top \eta_1\sr(w^n,x^n)$ satisfy
\begin{equation*}
|\tilde\chi^e{}^n(t_n)|+|\tilde\mu^n(t_n)|\le 4 q_4 g^{1-d} \bar\delta + \varepsilon
\end{equation*}
for all $n\ge N(\varepsilon)$. Let $\bar e:=h_e(\bar w,\bar x)$ denote the regulation error computed at $(\bar w,\bar x,\bar \eta)$. Then for each $n\in\N$ it satisfies
\begin{align*}
|\bar e| =   |C_e\tilde\chi^n(t_n) -C_eC_e^\top \tilde\mu_1^n(t_n)| \le |\tilde \chi^e{}^n(t_n)|+|\tilde\mu^n(t_n)|,
\end{align*} 
so that, for all $\varepsilon>0$, by taking $n\ge N(\varepsilon)$, we obtain
\[
|\bar e|\le 4 q_4 g^{1-d} \bar\delta + \varepsilon.
\]
For the arbitrariness of $\varepsilon$ and $(\bar w,\bar x,\bar \eta)\in\cA$ we thus conclude that
\begin{equation*}
|e|=|h_e(w,x)|\le 4 q_4 g^{1-d}\bar\delta
\end{equation*}
holds for all $(w,x)\in\cO$. Then the claim \eqref{e:bound_e} of the theorem follows from the  continuity of $h_e$ and uniform attractiveness of $\cA$ from $W\x X\x H$ by assuming, without loss of generality, $g>\sqrt[d-1]{4q_4/\varepsilon}$.  \endproof

%%%%%%%%%%%%%%%%%%%%%%%%%%%%%%%%%%%%%%%%
%

%\section{Proof of Proposition 1}\label{sec:proof}
%We develop the proof in 3 sections. In the first we prove uniform boundedness of the closed-loop system (i.e. point 1 of the proposition), in the second we prove that the trajectories are uniformly attracted by an $\varepsilon$ neighborhood of $\cO$ (i.e. point 2 of the proposition) and in the third section we prove the bound \eqref{1n:e:lime} (that is point 3 of the proposition).\\[1em]
%\textbf{I. Uniform boundedness:}\\
%\\
%[1em]
%%
%%
%\noindent\textbf{II. Existence of a Steady State:}\\
%Equation \eqref{e:bound_all_C} implies that, for each $H\subset\R^{\deta}$ compact, the $\Omega$-limit set $\Omega_{\Sigma_{cl}}(W\x X\x H)$ of the closed-loop system \eqref{s:wx}, \eqref{s:eta} is compact, non-empty, uniformly attractive from $W\x X\x H$ and invariant. Moreover, we can chose $H$ and $\ell$ such that $\cD:=\Omega_{cl}(W\x X\x H)\subset W\x X\x H$, so that $\cD$ is also asymptotically stable. Furthermore,   in view of \eqref{e:bound_all_C}, given any $\varepsilon>0$, then for large enough $\ell$ (say $\ell>\ell^\star_4(g)\ge \ell^\star_3(g)$) we have 
%\[
%|(w(t),x(t),\eta(t))|_\cC \le \beta_{\cC}(|(w(0),x(0),\eta(0))|_\cC,t) + \varepsilon,
%\]
%so that also point 2 of the proposition holds. \\[1em]
%%
%%
%\noindent\textbf{III. Asymptotic Bound:} \\

%%%%%%%%%%%%%%%%%%%%%%%%%%%%%%%%%%%%%%%%%%%%%%%%%%%%%%%%%%%%%%%

\section{On the Controllability Assumptions}\label{sec:P}
In this section, we complement Assumptions  \ref{ass:P_L} and \ref{ass:Dee}, by showing how they are implied by many  state-of-art assumptions routinely used in the context of high-gain stabilization and regulation of multivariable systems, and thus also showing how the matrix $\cL$ can be constructed in the respective frameworks  using only quantities that are known. In the following, we assume that $B(w,x)$ is $C^1$ and, for ease of notation, we let $\xb:=(w,x)$.  
\subsection{Strong Invertibility in the Sense of Wang and Isidori 2015 Implies Assumption \ref{ass:P_L}} \label{sec:P:1}
Here we prove that the assumption of invertibility used, for instance, in  recent papers by Wang, Isidori and Su~\cite{Wang2015} and by Wang, Isidori, Su and Marconi~\cite{Wang2016}  implies Assumption \ref{ass:P_L}.
\begin{lemma} \label{lem:P:1}
	Suppose that $\du=\dy$ (i.e., $B(\xb)$ is square), $B(\xb)$ is bounded, $L_{g(\xb)u}^{(x)}B(\xb)=0$ for all $\xb \in  W\x\R^\dx$, and there exists $\epsilon>0$ such that all its principal minors $\Delta_i(\xb)$, $i=1,\dots,r$ satisfy
	\begin{equation}\label{e:pp_Delta_i_ge_eps}
	|\Delta_i(\xb)|\ge \epsilon,
	\end{equation}   for all $\xb\in W\x\R^\dx$. Then,  Assumption \ref{ass:P_L} holds.
\end{lemma}
\begin{proof}
It can be shown~\cite{Wang2015} that if \eqref{e:pp_Delta_i_ge_eps} holds and $B(\xb)$ is bounded, then $B(\xb)$ can be written as \[B(\xb)= EM(\xb)(I+U(\xb)),\] with $M(\xb)=M(\xb)^\top$ positive definite for all $\xb\in  W\x\R^\dx$ and satisfying $a_1 I\le M(\xb)\le a_2 I$ for all $\xb\in W\x\R^\dx$ and for some $a_1,\,a_2>0$, $U(\xb)$ a strictly upper triangular matrix and with $E$ a diagonal matrix satisfying $EE=I$. Moreover, simple computations~\cite{Wang2015}  show that, if $B(\xb)$ is bounded, then there exists $c>1$ such that, with $C:=\diag(c^{\dy-1},c^{\dy-2},\dots,c,1)$, we have
	\[
	(I+U(\xb))C+C(I+U(\xb))^\top \ge I.
	\]
	Let
	\begin{align*}
	\cL&:=EC, & \cP(\xb)&:=EM(\xb)\inv E.
	\end{align*}
	Since   $a_1I \le M(\xb)  \le a_2 I$ holds at each $\xb\in W\x\R^\dx$, then the eigenvalues of $M(\xb)\inv$ are lower and upper-bounded by $a_2\inv$ and $a_1\inv$ respectively, and hence point a of Assumption \ref{ass:P_L} holds. Furthermore, we observe that  $\cL=\cL^\top =EC=CE$. Hence, noting that 
	\begin{equation*}
	B(\xb)^\top EM(\xb)\inv= (I+U(\xb))^\top M(\xb)^\top E^\top EM(\xb)\inv = (I+U(\xb))^\top, 
	\end{equation*}
	 we then have  
	\begin{align*}
	\cL^\top  B(\xb)^\top \cP(\xb)+\cP(\xb)B(\xb)\cL &= ECB(\xb)^\top  E M(\xb)\inv E + E M(\xb)\inv E B(\xb) CE= 
	E\Big[C (I+U(x))^\top   +  (I+U(x))C\Big]E \ge I,
	\end{align*}
	for all $\xb\in W\x\R^\dx$. Thus, since $EE=I$, point c of Assumption \ref{ass:P_L} holds. 
	
	Finally, we observe that~\cite{Wang2015} $L_{g(\xb)u}^{(x)}B(\xb)=0$ implies $L_{g(\xb)u}^{(x)}M(\xb) = 0$. Since 
	\[L_{g(\xb)u}^{(x)}M(\xb)\inv = -M(\xb)\inv L_{g(\xb)u}^{(x)}M(\xb)M(\xb)\inv=0,\] 
	then also point b holds, hence the result.
\end{proof} 

\subsection{Strong invertibility in the sense of Wang et al. 2015 and 2017  implies Assumption \ref{ass:P_L} } 
Here we prove that the assumption of invertibility used, for instance, in the works of Wang et al.~\cite{Wang2015uncB,Wang2017} implies Assumption~\ref{ass:P_L}.
\begin{lemma} 
	Suppose that $\du=\dy$ (i.e., $B(\xb)$ is square) and that there exist a non singular matrix $M\in\R^{\du\x \du}$ and a constant $\delta_0\in(0,1)$ such that
	\begin{equation}\label{e:pp_apdx:lemmaA52:hyp}
	\max_{ \substack{\Lambda\in\R^{\du\x \du}\\|\Lambda|\le 1}} \left| \Big(B(\xb)-M\Big)\Lambda M\inv \right|\le \delta_0
	\end{equation}
	holds for all $\xb\in W\x\R^\dx$. Then, Assumption \ref{ass:P_L} holds.
\end{lemma}
\begin{proof}
	As \eqref{e:pp_apdx:lemmaA52:hyp} holds for all $\Lambda\in\R^{\du\x\du}$ satisfying $|\Lambda|\le 1$, it holds in particular for $\Lambda=I$, thus yielding
	\[
	|B(\xb) M\inv - I |\le \delta_0.
	\]
	Thus, for all $\pb\in\R^\du$ and $\xb\in W\x\R^\dx$, it holds that
	\begin{equation*}
	2\pb^\top \Big( I-B(\xb)M\inv \Big)\pb  \le |2\pb^\top (B(\xb)M\inv-I)\pb|\le 2|\pb|^2\cdot |I-B(\xb)M\inv| \le 2\delta_0 |\pb|^2 =\pb^\top (2\delta_0 I)\pb.
	\end{equation*}
	Therefore, we obtain
	\begin{equation*}
	\pb^\top \Big( 2I-M^{-T}B(\xb)^\top -B(\xb)M\inv  \Big)\pb  = 2\pb^\top \pb - 2\pb^\top B(\xb)M\inv \pb  = 2\pb\Big(I-B(\xb)M\inv\Big)\pb \le \pb^\top (2\delta_0 I)\pb. 
	\end{equation*}
	As it holds for all $\pb\in\R^\du$ and $\xb\in W\x\R^\dx$ then, necessarily
	\[
	2I - M^{-T}B(\xb)^\top -B(\xb)M\inv \le 2\delta_0 I,
	\]
	and  thus, letting $\cL:=M\inv$ and $\cP(\xb):=I/(2(1-\delta_0))$ yields the item c of Assumption \ref{ass:P_L}. Point a instead follows by noting that $\delta_0\in(0,1)$, while point $b$ is straightforward as $\cP$ is constant in $\xb$. 
\end{proof}

\subsection{Positivity and negativity in the sense of McGregor, Byrnes and Isidori 2006, Back 2009, and Astolfi, Isidori, Marconi and Praly 2013 imply Assumption \ref{ass:P_L}}
Finally, here we show that the ``negativity'' and ``positivity'' assumptions on $B$, made in some recent papers~\cite{McGregor2006,Back2009,Astolfi2013}, all imply Assumption \ref{ass:P_L}. 
The following lemma   refers to the positivity assumption (Assumption 1) of~\cite{Astolfi2013}.
\begin{lemma}\label{lem:pp_pos_cond}
	Suppose that $\du=\dy$ (i.e. $B(\xb)$ is square) and that there exists $K\in\R^{\dy\x \dy}$ such that the following positivity condition holds 
	\begin{equation}\label{e:pp_pos_cond}
	B(\xb)K+K^\top B(\xb)^\top  \ge  I
	\end{equation}
	for all $\xb\in W\x\R^\dx$. Then Assumption \ref{ass:P_L} holds.
\end{lemma}  
\begin{proof}
We first observe that equation \eqref{e:pp_pos_cond} implies that $K$ is invertible. In fact, if we suppose that $K$ is singular, then there exists $\pb\ne 0$ satisfying $K\pb=0$. However, this yields $\pb^\top(B(\xb)K+K^\top B(\xb))\pb = 0$, which contradicts \eqref{e:pp_pos_cond}. 
%	Pre-multiplying by $K^{-\top}$ and post-multiplying by $K\inv$ \eqref{e:pp_pos_cond} yields
%\begin{equation*}
%K^{-\top} B(\xb)^\top + B(\xb) K\inv \ge K^{-\top}K\inv \ge \kappa I
%\end{equation*}
%in which $\kappa>0$ is the smallest singular value of $K\inv$. 
Then, Assumption \ref{ass:P_L} is satisfied by simply letting  $\cL=K$  and $\cP(\xb)=I$.  
\end{proof}
The following lemma, instead, refers to a slightly relaxed version of Assumption 4.4 of~\cite{McGregor2006}, which is, however, implied by that assumption in each compact subset of $W\x\R^\dx$.
\begin{lemma}\label{lem:pp_neg_cond}
	Suppose that $\du=\dy$ (i.e. $B(\xb)$ is square) and that there exists $M\in\R^{\dy\x \dy}$ and $\kappa>0$ such that the following negativity condition holds 
	\begin{equation}\label{e:pp_neg_cond}
	B(\xb)M+M^\top B(\xb)^\top  < -\kappa I
	\end{equation}
	for all $\xb\in W\x\R^\dx$. Then, Assumption \ref{ass:P_L} holds.
\end{lemma} 
\begin{proof}
	The proof follows by Lemma \ref{lem:pp_pos_cond} by noticing that \eqref{e:pp_neg_cond} implies \eqref{e:pp_pos_cond} with $K=-M/\kappa$.
\end{proof}
Finally, the following lemma concerns a slightly relaxed version of Assumption 3 of~\cite{Back2009}, which is equivalent to that assumption in each compact set.
\begin{lemma} Suppose that $\du=\dy$ (i.e. $B(\xb)$ is square), and assume that there exist $\kappa>0$, a non singular matrix $K$, $G^-:=\diag(g_1^-,\dots,g_{m}^-)$, and $G^+:=\diag(g_1^+,\dots,g_m^+)$ such that $0<G^-<G^+$ and that
	\begin{equation}\label{e:pp_back_c}
	\Big(B(\xb) K \pb-G^-\pb\Big)^\top \Pi^2\Big(B(\xb)K\pb-G^+\pb\Big)\le -\kappa \pb^\top\pb,
	\end{equation}
	for all $\pb\in\R^{\du}$ and all $\xb\in W\x\R^\dx$ and where $\Pi:=2(G^++G^-)\inv$. Then Assumption \ref{ass:P_L} holds.
\end{lemma}
\begin{proof}
	Equation \eqref{e:pp_back_c} implies ($G^-=(G^-)^\top $)
	\[
	-K^\top B(\xb)^\top  \Pi^2 G^+ - G^-\Pi^2 B(\xb) K+ K^\top B(\xb)^\top  \Pi^2 B(\xb)K + G^- \Pi^2 G^+   \le -\kappa I,
	\]
	that in turn implies
	\begin{equation*}
	\cM(\xb):= K^\top B(\xb)^\top  \Pi^2 G^+ + G^-\Pi^2 B(\xb) K > \kappa I
	\end{equation*}
	for all $\xb\in W\x\R^\dx$. Noting that 
	\begin{align*}
	\Pi^2 G^+ &= \Pi\cdot 2 (G^++G^-)\inv G^+ =  \Pi\cdot 2 (G^++G^-)\inv (G^++G^--G^-) = 2\Pi - \Pi^2 G^-\\
	G^- \Pi^2   &= 2\cdot G^-(G^++G^-)\inv \Pi = 2\Pi - G^+\Pi^2 
	\end{align*}
	then
	\begin{equation*}
	\cM(\xb) = 2 \Big(K^\top B(\xb)^\top \Pi + \Pi B(\xb)K\Big) - \cM(\xb)^\top 
	\end{equation*}
	Thus, $\cM(\xb)>\kappa I$ implies
	\[
	K^\top B(\xb)^\top \Pi + \Pi B(\xb)K = \dfrac{1}{2} \Big( \cM(\xb)+\cM(\xb)^\top \Big) >\kappa I,
	\]
	and by letting  $M= -K\Pi\inv$ we obtain
	\[
	M^\top B(\xb)^\top + B(\xb) M <-\kappa \Pi^2,
	\]
	and the claim follows from Lemma \ref{lem:pp_neg_cond} by noting that $\Pi^2$ is positive definite.
\end{proof}

\subsection{About Assumption \ref{ass:Dee}}
In this section, we show that Assumption \ref{ass:Dee} holds in all the frameworks reported in the previous sections. Regarding the framework of Section \ref{sec:P:1}, we observe that if $B(\xb)$ is bounded, $L_{g(\xb)u}^{(x)} B(\xb)=0$, and $|\Delta_i(\xb)|\ge \epsilon>0$ for all $\xb\in W\x\R^\dx$ and all $i=1,\dots, \dy$, then the same properties apply for the upper-left submatrix $B^{e,e}(\xb)\in\R^{\de\x \de}$ defined by eliminating the last $\dy-\de$ rows and columns of $B(\xb)$. Thus, according to Lemma \ref{lem:P:1}, there exist $\cL^{e}$ and $\cP^{e}(\xb)$ satisfying the same properties of Assumption \ref{ass:P_L} for the reduced matrix $B^{e,e}(\xb)$. Since by Lemma \ref{lem:P:1} the result for the full-size $B(\xb)$ can be obtained with a matrix $\cL$ of the form $\cL=\diag(c^{\dy-\de}\cL^{e},\cL^{a})$, for some $c>1$ and some $\cL^{a}\in\R^{(\dy-\de)\x (\dy-\de)}$, then in view of \eqref{d:D_K}, $D^{e,e}(\xb)=c^{\de-\dy}\cL^{e} B^{e,e}(\xb)$, and it is readily seen that $\cL$ can be chosen so that Assumptions \ref{ass:P_L} and \ref{ass:Dee} hold simultaneously with $\cM(\xb)=\cP^{e}(\xb)$.

Regarding all the other frameworks presented above, the fact that also Assumption \ref{ass:Dee} holds is a consequence of the following Lemma.
\begin{lemma}
Suppose that $\du=\dy$ and Assumption \ref{ass:P_L} holds  $\cP$   of the form\begin{align*}
\cP(\xb)&= \diag(\cP^e(\xb),\cP^a(\xb))
\end{align*}
for some $\cP^e(\xb)\in\R^{\de\x \de}$ and $\cP^a(\xb)\in\R^{(\dy-\de)\x(\dy-\de)}$. Then, Assumption \ref{ass:Dee} holds.
\end{lemma}
\begin{proof}
The proof directly follows by noticing that points a and b of Assumption \ref{ass:P_L} imply the same properties on $\cP^e(\xb)$, while point c, which by definition of $D(\xb)$ is equivalent to
\begin{equation*}
D(\xb)^\top \cP(\xb)+\cP(\xb)D(\xb) \ge I,
\end{equation*}
implies
\begin{equation}\label{e:P:in1}
\begin{pmatrix}
D^{e,e}(\xb)^\top\cP^e(\xb) + \cP^e(\xb) D^{e,e}(\xb)  & M_1(\xb) \\
M_1(\xb)^\top & M_2(\xb)
\end{pmatrix} \ge I
\end{equation}
for some properly defined $M_1$ and $M_2$. Pick $\pb\in\R^{\de}$ arbitrary. Then, \eqref{e:P:in1} implies
\begin{equation*} 
  |\pb|^2\le \begin{pmatrix}
\pb\\ 0
\end{pmatrix}^\top \begin{pmatrix}
D^{e,e}(\xb)^\top\cP^e(\xb) + \cP^e(\xb) D^{e,e}(\xb)  & M_1(\xb) \\
M_1(\xb)^\top & M_2(\xb)
\end{pmatrix}\begin{pmatrix}
\pb\\ 0
\end{pmatrix}  = \pb^\top \Big(D^{e,e}(\xb)^\top\cP^e(\xb) + \cP^e(\xb) D^{e,e}(\xb)\Big)\pb 
\end{equation*}
and thus Assumption \ref{ass:Dee} holds with $\cM(\xb):=\cP^e(\xb)$.
\end{proof}

\section{Conclusion}
In this paper, we have proposed a postprocessing regulator  for multivariable nonlinear systems possessing a partial normal form. Contrary to previous approaches, the proposed regulator can  handle   additional measurements that need not to vanish at the steady state but that can be useful for stabilization purposes or, for instance, to fit into the minimum-phase requirement. The proposed approach can also handle a control input dimension of arbitrarily large size, provided that a controllability assumption is fulfilled. About the controllability requirement, we have shown that it is implied by many state-of-art assumption in the literature of high-gain stabilization and regulation of normal forms. Among the drawbacks of the proposed approach, we underline how the stabilization phase is still strongly based on a   ``high-gain'' perspective, thus allowing us to restrict to linear stabilizers, and how the proposed conditions to obtain asymptotic regulation may be not constructive in general. 

Future research directions spread in many ways. In the first place, we aim to introduce adaptation in the control loop, so as to cope with model uncertainties and to weaken the chicken-egg dilemma that links the choice of the stabilizer gains and the internal model dynamics, as envisioned in~\cite{Bin2018b}. 
A second research direction concerns the extension of the framework to include more general stabilizing laws, going beyond the linear high-gain technique and trying to relax the minimum-phase requirement.

%%%%%%%%%%%%%%%%%%%%%%%%%%%%%%%%%%%
\appendix
\subsection{Proof of Lemma 1}\label{apd:lem1} 
	Pick $i\in \{1,\dots,r\}$ and, with $k_i>0$, define the matrix $\Lambda_i(k_i):=\diag( k_i^{N_i-2}I_{p_i},k_i^{N_i-3}I_{p_i},\dots,k_iI_{p_i},I_{p_i})\in\R^{p_i(N_i-1)\x p_i(N_i-1)}$ and change coordinates as
	\begin{equation*}
	\chi^i\mapsto  z^i:= \Lambda_i(k_i) \chi^i\,.
	\end{equation*}
	With reference to the matrices defined in \eqref{s:xi}-\eqref{s:y_xizeta}, noting that
	\begin{align*}
	 \Lambda_i(k_i) F_{ii}\Lambda_i(k_i)\inv &= k_i F_{ii},&
	\Lambda_i(k_i) H_{ii} &= H_{ii},&
	 \Lambda_i(k_i) C_i^\top  G_1 C_i \Lambda_i(k_i)\inv&= C_i^\top  G_1 C_i, & 
	 \Lambda_i(k_i) C_i^\top &=k_i^{N_i-2} C_i^\top 
	\end{align*}
	then $z^i$ fulfills
	\begin{equation*} 
	\dot{z}^i  = k_i  F_{ii}z^i +   H_{ii}\zeta_i + \Lambda_i(k_i)\sum_{j=1}^{i-1}\big( F_{ij}\Lambda_j(k_j)\inv z^j + H_{ij}\zeta_j  \big)  +  C_i^\top G_1C_i z^i + k_i^{N_i-2} C_i^\top  (\eta_2-G_1\eta_1) 
	\end{equation*}
	Let $(\alpha^i_1,\dots,\alpha^i_{N_i-1})$ be such that the polynomial $\lambda^{N_i-1}+\alpha_{N_i-1}^i\lambda^{N_i-2}+\dots+\alpha_{1}^i$ has  roots with negative real part and, with 
	\[D^i := \begin{pmatrix}
	-\alpha_1^i I_{p_i} & \cdots -\alpha_{N_i-1}^iI_{p_i}
	\end{pmatrix}\]
	and $\bar\zeta^i\in\R^{p_i}$, let 
	\begin{equation*}
	\zeta^i = \bar \zeta^i + k_i D^i z^i \,.
	\end{equation*}
	Since $F_{ii}+H_{ii}D^i$ is Hurwitz, then there exists $k^\star>1$ such that, for any $k_i>k^\star$, the following holds
	\begin{equation}\label{e:lem_bound_txi}  
	|z^i(t)|\le b_1 e^{-b_2k_it}|z^i(0)| +  b_3k^{N_i-3}_i\Big(  |\bar\zeta_{[1,i]} |_{[0,t)}+ |\upsilon|_{[0,t)}\Big)  + b_4k^{N_i-3}_i\sum_{j=1}^{i-1}   k_j\int_0^te^{-b_2k_i(t-\tau)}  |z^j(\tau)|  d\tau  
	\end{equation}
	for some $b_1,b_2,b_3,b_4>0$, and in which we let for convenience $\upsilon:=|G_1||\eta_1|+|\eta_2|$, $\bar\zeta_{[1,i]}:=(\bar\zeta_1,\dots,\bar\zeta_i)$. We can partition $\bar e$ as $\bar e=\col(\bar e^1,\dots,\bar e^{r_e})$, with  $\bar e^i:=k_i^{2-N_i}C_i z^i$. Pick $\epsilon>0$ and pick $i\in\{2,\dots,r\}$. Let for convenience  $\bar k_i:=\max_{1\le j<i}k_j$ and suppose that, for each $j=1,\dots,i-1$, $z^j(t)$ fulfills
	\begin{equation}\label{e:iss_j}
	|z^j(t)|\le h_1^i(\bar k_i)e^{-\frac{b_2}{2}  k_1 t}|z(0)|  + h_2^i(\bar k_i)\Big(|\bar\zeta_{[1,j]} |_{[0,t)} + |\upsilon|_{[0,t)}\Big),
	\end{equation}
	for some $h_1^i(\bar k_i),h_2^i(\bar k_i)\ge 0$. Then, if $k_i>\bar k_i>k_1$, for each $\tau\le t$ we have
		\begin{equation*}
		e^{-b_2 k_i (t-\tau)}\le e^{-b_2 k_1(t-\tau)}\le e^{-\frac{b_2}{2}k_1 t} e^{b_2 k_1\tau},
		\end{equation*}	
and  equations \eqref{e:lem_bound_txi}-\eqref{e:iss_j} give 
	\begin{equation}\label{e:pp_bound_zi}
	\begin{aligned} 
	|z^i(t)|&\le b_1 e^{-b_2  k_1t}|z(0)| +  b_3k^{N_i-3}_i\Big( |\bar\zeta_{[1,i]} |_{[0,t)}+ |\upsilon|_{[0,t)}\Big)\\&\qquad\qquad  + b_4k^{N_i-3}_i(i-1)\bar k_i\bigg( e^{-b_2 k_1 t} \int_0^te^{\frac{b_2}{2} k_1 \tau}  h_1^i(\bar k_i)|z(0)|d\tau  + \int_0^te^{-b_2k_i(t-\tau)} h_2^i(\bar k_i)\big(|\bar\zeta_{[1,i]} |_{[0,t)}+ |\upsilon|_{[0,\tau)} \big) d\tau \bigg)  
	\\&\le \left( b_1 + \dfrac{ 2b_4 r \bar k_i h_1^i(\bar k_i) k_i^{N_i-3} }{b_2 k_1}\right)  e^{-\frac{b_2}{2}  k_1t} |z(0)| + \left( b_3 + \dfrac{ b_4 r \bar k_i h_2^i(\bar k_i) }{b_2k_i}\right) k_i^{N_i-3}\big(|\bar\zeta_{[1,i]} |_{[0,t)}+|\upsilon|_{[0,t)}\big).
	\end{aligned}
	\end{equation}
	Hence, if
	\begin{equation}\label{d:lem_ki}
	k_i >\max\left\{ k^\star,\; \bar k_i,\;\dfrac{r}{\epsilon}\left( b_3+ \dfrac{b_4r\bar k_i h_2^i(\bar k_i)}{b_2}\right)\right\},
	\end{equation}
	then, in view of \eqref{e:pp_bound_zi}, the fact that \eqref{e:iss_j} holds for $j=1,\dots,i-1$ implies that the same bounds also hold  for $j=1,\dots,i$, with $\bar k_{i+1}=\max\{\bar k_i,k_i\}=k_i$ and 
	\begin{align*}
	h_1^{i+1}(\bar k_{i+1})&:=\  h_1^i(\bar k_{i+1})+ b_1 + \dfrac{ 2b_4 r   h_1^i(\bar k_{i+1}) \bar k_{i+1}^{N_i-2} }{b_2k_1} ,&
	h_2^{i+1}(\bar k_{i+1}) & :=   h_2^i(\bar k_{i+1})+\dfrac{r} {\epsilon} \bar k_{i+1}^{N_i-2} .
	\end{align*}
	Moreover, noting that  $|\bar e^i(t)|\le k_i^{2-N_i}|z^i(t)|$, in view of \eqref{e:pp_bound_zi},  we have
	\begin{equation}\label{ineq:lem_tilde_e_i}
	|\bar e^i(t)|  \le  q_1^i(\bar k_i,k_i) e^{-\frac{b_2}{2} k_1 t}|z(0)| + \dfrac{\epsilon}{r} \big(|\bar\zeta^e|_{[0,t)}+|\upsilon|_{[0,t)}\big)
	\end{equation}
	with $q_1^i(\bar k_i,k_i)  := h_1^{i+1}(\bar k_i) k_i^{2-N_i}$  and $\bar\zeta^e:=\bar\zeta_{[1,r_e]}$. Fix $k_1$ so that
	\begin{equation*}
	k_1 \ge \max\left\{ k^\star, \dfrac{b_3 r}{\epsilon}\right\}\,.
	\end{equation*}
	Since $e^{-b_2k_1t}\le e^{-b_2 k_1 t/2}$ for all $k_1,t>0$ and $\bar k_2=k_1$, then \eqref{e:lem_bound_txi} implies \eqref{e:iss_j} for $j<2$ with $h_1^2(\bar k_2)=b_1$ and $h_2^2(\bar k_2)= b_3 \bar k_2^{N_1-3}$  and \eqref{ineq:lem_tilde_e_i} for $i=1$, with   $q_1^1:=b_1k^{2-N_1}$, $q_2^1:=b_2k_1$.  
	 Hence the sequence \eqref{d:lem_ki} is well-defined and by induction we conclude that, for each $i=1,\dots, r$, it holds that
	\begin{align*}
	|z^i(t)|&\le \bar h_1(\bar k_{i+1}) e^{-\frac{b_2}{2} k_1 t}|z(0)| + \bar h_2(\bar k_{i+1}) \Big(|\bar\zeta_{[1,i]} |_{[0,t)} + |\upsilon|_{[0,t]}\Big)
	\end{align*} 
		with $\bar h_1 := h_1^{r+1}$, $\bar h_2:= h_2^{r+1}$, and for each $i=1,\dots,r_e$, \eqref{ineq:lem_tilde_e_i} holds true.
		Noting that $|\chi^i(t)|\le |z^i(t)|$, $|z(0)|\le k^{N_{r}-1}|\chi(0)|$ and $|\bar e(t)|\le \sum_{i=1}^{r_e} |\bar e^i(t)|$, then we obtain \eqref{e:lemma_k_claim_iss}-\eqref{e:lemma_k_claim_oiss}, with
	\begin{align*}
	\beta_\chi(s,\tau) &:= k^{N_{r}-1} \bar h_1(\bar k_{r+1})  \exp\left(-\frac{b_2 b_3 r  \tau}{2}\right) s, & a_1&:= k^{N_{r}-1} \bar h_2(\bar k_{r+1}) ,& \beta_{\bar e}(s,t/\epsilon) &:=s\sum_{i=1,\dots, r_e}  q_1^i(\bar k_i,k_i) \exp\left(-\frac{b_2 b_3 r  \tau}{2}\right) 
	\end{align*} 
	and the claim of the lemma follows with
	\begin{equation*}%\label{m:d:lem_K}
	K := \diag\Big(
	k_1 D^1 \Lambda_1(k_1),  \  \ldots, \  k_{r} D^{r} \Lambda_{r}(k_{r})\Big).
	\end{equation*} \endproof

\subsection{Proof of Lemma 2}
With $\cP$ given by Assumption \ref{ass:P_L},  define the function
\begin{equation}
V(w,x) = \sqrt{\bar\zeta^\top \cP(w,x)\bar\zeta }
\end{equation}
By point a of Assumption \ref{ass:P_L} $V$ satisfies
\[
\sqrt{\und\lambda} |\bar\zeta| \le V(w,x)\le \sqrt{\bar\lambda}|\bar\zeta|
\]
for all $(w,x)\in W\x \R^\dx$. 
Taking the Dini derivative of $V$ along the solutions of the closed-loop system yields
\begin{align*}
D^+V(w,x) = \dfrac{1}{2V(w,x)} \Big(&-\ell \bar\zeta^\top \Big(\cL^\top B(w,x)^\top \cP(w,x)+ \cP(w,x)B(w,x)\cL  \Big)\bar\zeta\\
&+2\bar\zeta^\top \cP(w,x)\big(\varrho(\eta,\chi,\bar\zeta)+q(w,x)\big) \\&+ \bar\zeta\big( L_s^{(w)}\cP(w,x) + L_f^{(x)} \cP(w,x) +  L_{b(w,x)u}^{(x)}\cP(w,x) \big) \bar\zeta \Big).
\end{align*}
Point c of Assumption \ref{ass:P_L} implies
\begin{equation*}
\bar\zeta^\top \Big(\cL^\top B(w,x)^\top \cP(w,x)+ \cP(w,x)B(w,x)\cL  \Big)\bar\zeta \ge |\bar\zeta|^2
\end{equation*} 
so as 
\[
-\ell \bar\zeta^\top \Big(\cL^\top B(w,x)^\top \cP(w,x)+ \cP(w,x)B(w,x)\cL  \Big)\bar\zeta\le -\ell |\bar\zeta|^2.
\]
%Let $\bar q(w,x)$ be any function that agrees with $q$ on $\cA$ and that satisfies $|q(w,x)|\le \sup_{(w,x)\in\cA}|q(w,x)|$ for all $(w,x)\in\R^\dw\x\R^\dx$. Adding and subtracting $2\bar\zeta^\top \cP(w,x)\bar q(w,x)$ yields
%\begin{equation*}
%2 \bar\zeta^\top \cP(w,x)\big(\rho(\eta,\chi,\bar\zeta)+q(w,x)\big) 
%=2\bar\zeta^\top \cP(w,x)\big(\rho(\eta,\chi,\bar\zeta)+\tilde q(w,x) +\bar q(w,x)\big)  
%\end{equation*}
%with $\tilde q:=q-\bar q$. 
As $\cP$ is continuous, $\varrho$ is Lipschitz and $q$ is Lipschitz on $W\x \bar X$, then there exists $M_1>0$ such that
\begin{equation*}
2 \bar\zeta^\top \cP(w,x)\big(\varrho(\eta,\chi,\bar\zeta)+q(w,x)\big) \le M_1|\bar\zeta|\Big(  |\bar\zeta| + |(w,x,\chi,\eta)|  \Big),
\end{equation*}
as long as $(w,x)\in W\x \bar X$. Point b of Assumption \ref{ass:P_L} implies that $L_{b(w,x)u}^{(x)}\cP(w,x) =0$, so as by continuity of $\cP$, as long as $(w,x)\in W\x \bar X$, we can write
\begin{equation*}
\bar\zeta\big( L_s^{(w)}\cP(w,x) + L_f^{(x)} \cP(w,x) +  L_{b(w,x)u}^{(x)}\cP(w,x) \big) \bar\zeta\le M_2 |\bar\zeta|^2
\end{equation*}
for some $M_2>0$. Since
\[
\dfrac{1}{\sqrt{\bar\lambda}}\le \dfrac{|\bar\zeta|}{V(w,x)} \le \dfrac{1}{\sqrt{\und\lambda}} ,
\]
then there exists $\alpha_2>0$ such that, as long as $(w,x)\in W\x \bar X$, we have
\begin{equation*}
D^+V(w,x) \le (\alpha_2-\ell\alpha_1) V(w,x) + \alpha_2 |(w,x,\eta,\chi)| ,
\end{equation*}
with $\alpha_1:=1/(2\sqrt{\bar \lambda})$. The result thus follows with $\beta_{\bar\zeta}(s,\tau):= (s \sqrt{\bar\lambda}/\sqrt{\und\lambda}) \exp(-\alpha_1\tau/2)$ and $a_4:=2\alpha_2/(\alpha_1\sqrt{\und\lambda})$ by taking    $\ell^\star_0:= 2\alpha_2/\alpha_1$ and noticing that $\beta_{\bar \zeta}$ does not depend on $\bar X$.\endproof

\bibliographystyle{ieeetran}
\bibliography{bibliography}

\begin{thebibliography}{10}
\providecommand{\url}[1]{#1}
\csname url@samestyle\endcsname
\providecommand{\newblock}{\relax}
\providecommand{\bibinfo}[2]{#2}
\providecommand{\BIBentrySTDinterwordspacing}{\spaceskip=0pt\relax}
\providecommand{\BIBentryALTinterwordstretchfactor}{4}
\providecommand{\BIBentryALTinterwordspacing}{\spaceskip=\fontdimen2\font plus
\BIBentryALTinterwordstretchfactor\fontdimen3\font minus
  \fontdimen4\font\relax}
\providecommand{\BIBforeignlanguage}[2]{{%
\expandafter\ifx\csname l@#1\endcsname\relax
\typeout{** WARNING: IEEEtran.bst: No hyphenation pattern has been}%
\typeout{** loaded for the language `#1'. Using the pattern for}%
\typeout{** the default language instead.}%
\else
\language=\csname l@#1\endcsname
\fi
#2}}
\providecommand{\BIBdecl}{\relax}
\BIBdecl

\bibitem{Francis1976}
B.~A. Francis and W.~M. Wonham, ``The internal model principle of control
  theory,'' \emph{Automatica}, vol.~12, pp. 457--465, 1976.

\bibitem{Davison1976}
E.~J. Davison, ``The robust control of a servomechanism problem for linear
  time-invariant multivariable systems,'' \emph{IEEE Transactions on Automatic
  Control}, vol. AC-21, no.~1, pp. 25--34, 1976.

\bibitem{Byrnes2003}
C.~I. Byrnes and A.~Isidori, ``Limit sets, zero dynamics and internal models in
  the problem of nonlinear output regulation,'' \emph{IEEE Transactions on
  Automatic Control}, vol.~48, pp. 1712--1723, Oct. 2003.

\bibitem{Bin2018b}
M.~Bin and L.~Marconi, ``The chicken-egg dilemma and the robustness issue in
  nonlinear output regulation with a look towards adaptation and universal
  approximators,'' in \emph{IEEE 57th Conference on Decision and Control
  (CDC)}, Miami Beach, FL, USA, 2018.

\bibitem{Huang1994}
J.~Huang and C.~F. Lin, ``On a robust nonlinear servomechanism problem,''
  \emph{IEEE Transactions on Automatic Control}, vol.~39, no.~7, pp.
  1510--1513, 1994.

\bibitem{Byrnes1997}
C.~Byrnes, F.~{Delli Priscoli}, and A.~Isidori, ``Structurally stable output
  regulation for nonlinear systems,'' \emph{Automatica}, vol.~33, no.~3, pp.
  369--385, 1997.

\bibitem{Huang2001}
J.~Huang, ``Remarks on the robust output regulation problem for nonlinear
  systems,'' \emph{IEEE Transactions on Automatic Control}, vol.~46, no.~12,
  pp. 2028--2031, 2001.

\bibitem{Huang2004}
J.~Huang and Z.~Chen, ``A general framework for tackling the output regulation
  problem,'' \emph{IEEE Transactions on Automatic Control}, vol.~49, no.~12,
  pp. 2203--2218, 2004.

\bibitem{Isidori2012}
A.~Isidori, L.~Marconi, and L.~Praly, ``Robust design of nonlinear internal
  models without adaptation,'' \emph{Automatica}, vol.~48, pp. 2409--2419,
  2012.

\bibitem{Marconi2007}
L.~Marconi, L.~Praly, and A.~Isidori, ``Output stabilization via nonlinear
  {L}uenberger observers,'' \emph{SIAM Journal of Control and Optimization},
  vol.~45, pp. 2277--2298, 2007.

\bibitem{Bin2018c}
M.~Bin, D.~Astolfi, L.~Marconi, and L.~Praly, ``About robustness of internal
  model-based control for linear and nonlinear systems,'' in \emph{57th IEEE
  Conference on Decision and Control}, Miami Beach, FL, USA, 2018.

\bibitem{Byrnes2004}
C.~I. Byrnes and A.~Isidori, ``Nonlinear internal models for output
  regulation,'' \emph{IEEE Transactions on Automatic Control}, vol.~49, pp.
  2244--2247, Dec. 2004.

\bibitem{Chen2005}
Z.~Chen and J.~Huang, ``Robust output regulation with nonlinear exosystems,''
  \emph{Automatica}, vol.~41, pp. 1447--1454, 2005.

\bibitem{IsidoriBookNew}
A.~Isidori, \emph{Lectures in Feedback Design for Multivariable Systems}.\hskip
  1em plus 0.5em minus 0.4em\relax Springer International Publishing, 2017.

\bibitem{McGregor2006}
N.~K. McGregor, C.~I. Byrnes, and A.~Isidori, ``Results on nonlinear output
  regulation for {MIMO} systems,'' in \emph{2006 American Control Conference
  (ACC 2006)}, 2006.

\bibitem{Astolfi2013}
D.~Astolfi, A.~Isidori, L.~Marconi, and L.~Praly, ``Nonlinear output regulation
  by post-processing internal model for multi-input multi-output systems,'' in
  \emph{9th {IFAC} Symposium on Nonlinear Control Systems}, Tououse, France,
  2013, pp. 295--300.

\bibitem{Wang2016}
L.~Wang, A.~Isidori, H.~Su, and L.~Marconi, ``Nonlinear output regulation for
  invertible nonlinear {MIMO} systems,'' \emph{International Journal of Robust
  and Nonlinear Control}, vol.~26, pp. 2401--2417, 2016.

\bibitem{Wang2017}
L.~Wang, A.~Isidori, Z.~Liu, and H.~Su, ``Robust output regulation for
  invertible nonlinear {MIMO} systems,'' \emph{Automatica}, vol.~82, pp.
  278--286, 2017.

\bibitem{Wang2018}
L.~Wang, L.~Marconi, C.~Wen, and H.~Su, ``Pre-processing nonlinear output
  regulation with non vanishing measurements,'' in \emph{IEEE 57th Conference
  on Decision and Control}, 2018, pp. 5640--5645.

\bibitem{Isidori2012b}
A.~Isidori and L.~Marconi, ``Shifting the internal model from control input to
  controlled output in nonlinear output regulation,'' in \emph{Proceedings of
  the 51st IEEE Conference on Decision and Control (CDC 2012)}, Maui, Hawaii,
  USA, 2012, pp. 4900--4905.

\bibitem{Bin2017ifac}
M.~Bin and L.~Marconi, ``About a post-processing design of regression-like
  nonlinear internal models,'' in \emph{IFAC 2017 World Congress}, Toulouse,
  France, 2017.

\bibitem{Astolfi2015b}
D.~Astolfi, L.~Praly, and L.~Marconi, ``Approximate regulation for nonlinear
  systems in presence of periodic disturbances,'' in \emph{IEEE 54th Conference
  on Decision and Control (CDC 2015)}, 2015, pp. 7665--7670.

\bibitem{Astolfi2017}
D.~Astolfi and L.~Praly, ``Integral action in output feedback for multi-input
  multi-output nonlinear systems,'' \emph{IEEE Transactions on Automatic
  Control}, vol.~62, no.~4, pp. 1559--1574, 2017.

\bibitem{Bin2018CDC}
M.~Bin, D.~Astolfi, L.~Marconi, and L.~Praly, ``About robustness of internal
  model-based control for linear and nonlinear systems,'' in \emph{IEEE 57th
  Conference on Decision and Control (CDC)}, Miami Beach, FL, USA, 2018.

\bibitem{Isidori1995book2}
A.~Isidori, \emph{Nonlinear Control Systems II}.\hskip 1em plus 0.5em minus
  0.4em\relax Springer-Verlag, 1999.

\bibitem{Hirschorn1979}
M.~R. Hirschorn, ``Invertibility of multivariable nonlinear control systems,''
  \emph{IEEE Transactions on Automatic Control}, vol.~24, no.~6, pp. 855--865,
  1979.

\bibitem{Singh1981}
S.~N. Singh, ``A modified algorithm for invertibility of nonlinear systems,''
  \emph{IEEE Transactions on Automatic Control}, vol.~26, no.~2, pp. 595--598,
  1981.

\bibitem{Isidori1995book}
A.~Isidori, \emph{Nonlinear Control Systems. 3rd ed.}\hskip 1em plus 0.5em
  minus 0.4em\relax Springer-Verlag, 1995.

\bibitem{Liberzon2002}
D.~Liberzon, A.~S. Morse, and E.~D. Sontag, ``Output-input stabiity and
  minimum-phase nonlinear systems,'' \emph{IEEE Transactions on Automatic
  Control}, vol.~47, no.~3, pp. 422--436, 2002.

\bibitem{Liberzon2004}
D.~Liberzon, ``Output-input stability implies feedback stabilization,''
  \emph{Systems \& Control Letters}, vol.~53, no.~3, pp. 237--248, 2004.

\bibitem{Wang2015}
L.~Wang, A.~Isidori, and H.~Su, ``Global stabilization of a class of invertible
  {MIMO} nonlinear systems,'' \emph{IEEE Transactions on Automatic Control},
  vol.~60, no.~3, pp. 616--631, 2015.

\bibitem{Wang2015uncB}
------, ``Output feedback stabilization of {MIMO} systems having uncertain
  high-frequency gain matrix,'' \emph{Systems \& Control Letters}, vol.~83, pp.
  1--8, 2015.

\bibitem{Back2009}
J.~Back, ``An inner-loop controller guaranteeing robust transient performance
  for uncertain {MIMO} nonlinear systems,'' \emph{IEEE Transactions on
  Automatic Control}, vol.~54, no.~7, pp. 1601--1607, 2009.

\bibitem{Atassi1999}
A.~N. Atassi and H.~K. Khalil, ``A separation principle for the stabilization
  of a class of nonlinear systems,'' \emph{IEEE Trans. Autom. Contr.}, vol.~44,
  no.~9, 1999.

\bibitem{Teel1995}
A.~R. Teel and L.~Praly, ``Tools for semiglobal stabilization by partial state
  and output feedback,'' \emph{SIAM Journal on Control and Optimization},
  vol.~33, pp. 1443--1488, 1995.

\bibitem{Khalil2013}
H.~Khalil and L.~Praly, ``High-gain observers in nonlinear feedback control,''
  \emph{International Journal of Robust and Nonlinear Control}, vol.~24, no.~6,
  pp. 993--1015, 2013.

\bibitem{Jiang1994}
Z.-P. Jiang, A.~R. Teel, and L.~Praly, ``Small-gain theorem for {ISS} systems
  and applications,'' \emph{Mathematics of Control, Signals, and Systems},
  vol.~7, pp. 95--120, 1994.

\end{thebibliography}

\end{document}